\documentclass[twocolumn, pra,superscriptaddress, amsmath,amssymb, aps,longbibliography]{revtex4-2}

\usepackage{amsthm}%
\usepackage{mathrsfs}%
\usepackage[title]{appendix}%
\usepackage{xcolor}%
\usepackage{textcomp}%
\usepackage{subfigure}

\usepackage{booktabs}%
\usepackage{algorithm}%
\usepackage{algorithmicx}%
\usepackage{algpseudocode}%
\usepackage{listings}%
\usepackage{bm}
\usepackage{bbold}
\usepackage{color}
\usepackage[english]{babel}
\usepackage{lipsum}
\usepackage[normalem]{ulem}
\usepackage{grffile}
\usepackage{braket}
\usepackage{appendix}
\usepackage[colorlinks, linkcolor=blue, anchorcolor=blue, citecolor=blue, urlcolor=blue]{hyperref}

\newcommand{\bfu}{\mathbf{u}}
\newcommand{\bfv}{\mathbf{v}}
\newcommand{\bfx}{\mathbf{x}}
\newcommand{\bfy}{\mathbf{y}}
\newcommand{\bfa}{\mathbf{a}}
\newcommand{\bfb}{\mathbf{b}}
\newcommand{\bfc}{\mathbf{c}}

\newcommand{\bfk}{\mathbf{k}}
\newcommand{\bfl}{\mathbf{l}}
\newcommand{\bfZ}{\mathbb{Z}}
\newcommand{\comb}[2]{{C}^{#2}_{#1}}

\newcommand{\poly}{\mathrm{poly}}
\newcommand{\sgn}{\mathrm{sgn}}

\newcommand{\tr}[1]{\mathrm{tr}\!\left(#1\right)}

\newtheorem{lemma}{Lemma}
\newtheorem{theorem}{Theorem}
\newtheorem{proposition}{Proposition}%

\begin{document}
\title{Low-overhead tailoring and learning of noise in graph states}
\author{Guedong Park}
\affiliation{NextQuantum Innovation Research Center, Department of Physics and Astronomy, Seoul National University, Seoul, 08826, Korea}

\author{Jinzhao Sun}
\email{jinzhao.sun.phys@gmail.com}
\affiliation{School of Physical and Chemical Sciences, Queen Mary University of London, London E1 4NS, United Kingdom }

\author{Hyunseok Jeong}
\email{jeongh@snu.ac.kr}
\affiliation{NextQuantum Innovation Research Center, Department of Physics and Astronomy, Seoul National University, Seoul, 08826, Korea}

\begin{abstract}
Graph and hypergraph states are important resource states for realizing universal quantum computation and diverse non-local physical phenomena.
However, noise learning in such states is challenging due to their large entanglement and magic. This work establishes a low Clifford hierarchy circuit-based scheme for tailoring and learning noise in (hyper)graph states generated by multiple controlled-Z gates. 
The key is to convert the noisy input state into diagonal form and derive the convolution equation for the diagonal noise rates, which has a lower resource overhead. Such a diagonal form can be used in real quantum simulation by using the same conversion technique on the graph state input. We demonstrate that a single-depth Bell measurement is sufficient in our scheme for arbitrary graph states, while noise tailoring can be done via Pauli operations. After that, we suggest the Walsh-Hadamard transform and some approximation method for decoding the convolution equation to estimate the tailored noise. 
We prove that constant sampling and polynomial time complexities are sufficient to guarantee bounded noise estimation error with respect to the $l_2$-norm, which usually requires exponential complexity with existing noise estimation methods. The polynomial complexities are maintained even for the more stringent $l_1$ approximation under the sparse noise assumption. Conclusively, we propose novel methods for characterizing the noise properties of an entangled state at a lower cost than that of state generation. Moreover, compared with state verification methods, we can attain richer information on noise rates and enable noise mitigation, thereby providing guarantees for the preparation of graph states.

\end{abstract}
\maketitle

\section{Introduction}

Quantum entanglement~\cite{horodecki2009,chitambar2016} is a representative resource of quantum advantages over classical computing~\cite{bennett1993, bennett2014}. Graph states~\cite {lee2024}, more generally, the hypergraph states~\cite{zhu2019} can host rich entanglement and non-stabilizerness~\cite{bravyi2019,howard2017,chen2024}. They play an indispensable role in a broad range of applications, including measurement-based quantum computing (MBQC)~\cite{briegel2009,miller2016,takeuchi2019}, error-correcting code~\cite{google2023,fowler2012}, quantum metrology~\cite{shettell2020}, and fundamental studies on multipartite entanglement~\cite{yoshida2016,gachechiladze2016}. However, the preparation of such highly entangled resource states inevitably suffers from physical (and even logical~\cite{aharonov2025}) noise. Accordingly, noise learning~\cite{chen2023} for graph state inputs becomes an essential step toward developing improved graph state generation protocols~\cite{huet2025,thomas2024}. Once the underlying noise characteristics are identified, we can employ post-error mitigation~\cite{pivateau2021,suzuki2022,wallman2018}, or noise-adaptive error-correcting codes~\cite{layden2020,basak2025} to substantially reduce the noise of graph states. Moreover, the noise can be tailored into a more tractable form (e.g., Pauli noise)~\cite{harper2020,emerson2007}, facilitating analytical treatment and further optimization.


However, despite its importance, establishing efficient schemes for noise tailoring and learning remains an open challenge.
Existing noise tailoring techniques are generally limited to twirling the gate noise~\cite{emerson2005,wallman2016,liu2024}, making them difficult to apply across diverse graph-state generation platforms such as vertex fusion~\cite{lee2023} and distillation~\cite{litinski2019}.
Moreover, learning the noise characteristics is even more demanding than verification~\cite{li2023,tanizawa2023,huang2025_sv,zhu2019,morimae2017,takeuchi2018,donnell2025}, which typically does not estimate many noise parameters.
State or gate-set tomography~\cite{endo2018,pivateau2021,gross2010} can, in principle, reconstruct the full noise structure, but its sampling and time complexities scale exponentially with the number of qubits~\cite{gross2010}.
Alternative approaches for learning Pauli noise~\cite{harper2021,flammia2020,harper2020} utilize the Walsh–Hadamard transform (WHT) on the Pauli eigenvalue spectrum.
Nevertheless, these methods usually assume random sparse noise~\cite{harper2021} and its gate-independence during repetitive Pauli operations for cyclic benchmarking~\cite{flammia2020}, which limits their applicability in realistic graph-state platforms.

In this work, we propose a noise tailoring framework that diagonalizes a noisy multi-qubit (hyper)graph state and an efficient learning scheme for its diagonal components in the low-error regime. 
Learning these diagonal elements is sufficient for graph state-based quantum simulation, since the same tailoring process can be inserted into the noisy graph state inputs.
Our protocol achieves this by employing gates from a lower level of the Clifford hierarchy~\cite{anderson2024} than those used in the target-graph state generation, together with depth-one transversal \textsc{CNOT} operations.
Although the algorithm itself is assumed to be noiseless, it yields a definite improvement in physical overhead, requiring significantly fewer entangling gates compared to direct measurements in the target graph-state bases.

To do so, we prepare the Bell-like measurement~\cite{catani2018} circuit acting on two copies of the noisy graph states. This circuit transforms the noise of each copy into some dephasing noise and then obtains the final measurement outcome following the self-convolution of the dephasing rate.
We propose two strategies, the Walsh-Hadamard transform (WHT) and the convolution power method (CPM), for decoding the self-convolution to reveal the dephasing noise rate. 

Compared to the previous method~\cite{chen2023,flammia2020,harper2021}, we use WHT for the self-convolution of the target dephasing rate, which requires neither Pauli eigenvalue estimation nor cyclic benchmarking of the noise channel. The WHT allows us to reach the exact distribution of the dephasing noise, while its running time may be inefficient. On the other hand, the CPM approximates the noise by truncating the higher-order component in the perturbed expansion of the solution, and is valid for the low infidelity cases. We prove that constant sampling and polynomial time complexities, for a fixed hypergraph order, are required to guarantee the noise estimation with respect to the $l_2$-norm, regardless of the support size of the noise. It overcomes the complexity of the conventional WHT to Pauli eigenvalues~\cite{flammia2020,harper2021}, given relatively small noise rates. 

The polynomial efficiency is preserved even in the $l_1$ approximation, which is a much stricter requirement for noise learning, when the size of support in the tailored noise~\cite{harper2020,berg2023} is polynomially bounded by the number of qubits. This sparsity condition holds for many graph states, including cluster states~\cite{briegel2009,nielsen2005_fault_tolerant}, which have a constant number of controlled-$Z$ connections and are widely used as resource states for fault-tolerant quantum computation.
Our results provide guarantees for noise characterization in preparing graph states and yield deeper insights into the nature of noise in highly entangled systems.

\section{Preliminaries}

Throughout this paper, we consider $n$-qubit quantum system. 
$\mathcal{P}_n\equiv \left\{\pm iI,\pm iX,\pm iY,\pm iZ\right\}^{\otimes n}$ denotes the $n$-qubit Pauli group~\cite{gottesman1998}. Given a single-qubit Pauli operator $P$, we denote an $n$-qubit Pauli operator $P^{\bfa}\equiv \bigotimes_{i=1}^{n} P^{a_i}$, where $\bfa\in \bfZ^n_2$. An up-to-phase Pauli operator (Weyl-Heisenberg operator) in $\mathcal{P}_n/ \mathbb{Z}_4$ is denoted as $T_{\bfa}\equiv \bigotimes_{i=1}^{n}i^{a_{ix}a_{iz}}X^{a_{ix}}Z^{a_{iz}}$, where $\bfa=(\bfa_x,\bfa_z)\in \bfZ^{2n}_2$. Also, we define the $k(\in \mathbb{N})$-th Clifford hierarchy~\cite{anderson2024} ${\rm Cl}^{(k)}_n$ as the set of operators,
\begin{align}
    \mathrm{Cl}^{(k)}_n=\left\{U\in {\rm SU}(2^n)\,|\,\forall E\in \mathcal{P}_n,\,UEU^{\dagger}\in \mathrm{Cl}^{(k-1)}_n\right\},
\end{align}
and $\mathcal{P}_n=\mathrm{Cl}^{(1)}_n$. $\mathrm{Cl}^{(2)}_n$ is simply denoted as ${\rm Cl}_n$, called as Clifford group~\cite{aaronson2004,huang2020}. The stabilizer state is the pure state generated by $U\in {\rm Cl}_n$ operation to the computational basis. We can define the magic states~\cite{bravyi2005,bravyi2012} (or non-stabilizer states~\cite{leone2022}), which are outputs of non-Clifford gates to $\frac{1}{\sqrt{2^n}}\left(\ket{0}+\ket{1}\right)^{\otimes n}\;(\equiv \ket{+}^{\otimes n})$. 
Moreover, we define $k$th-order controlled-Z gate that is for the computational basis $\forall \ket{\bfx} (\bfx\in \bfZ^n_2)$, 
\begin{align}
    C_{\left\{i_1,i_2,\ldots,i_k\right\}}Z\ket{\bfx}=(-1)^{x_{i_1}x_{i_2}\ldots x_{i_k}}\ket{\bfx},
\end{align}
where $i_j\in [n]$, $j\in [k]$. We say it as the multiple controlled-Z gate if the specification of $k$ is unnecessary.  When $k=3$, it is a controlled-controlled-Z (\textsc{CCZ})-gate and is denoted  as $\textsc{CCZ}_{\left\{i_1,i_2,i_3\right\}}$. When $k=1~(2\;{\rm resp.})$, it is $Z$(\textsc{CZ})-gate. We now consider a hypergraph $G(V,E)$ where $V=[n]$ and $E$ are subsets in $V$ with a maximal size $k\ge 2$~\cite{bretto2013}. The $k$th-order hypergraph state $\ket{G(V,E)}$, which we shall denote $\ket{\psi}$ for brevity, is defined as~\cite{rossi2013}, 
\begin{align}\label{main:def_hypergraph_state}
    \ket{\psi}&\equiv \prod_{A\in E}C_A Z\ket{+}^{\otimes n}=\frac{1}{\sqrt{2^n}}\sum_{\bfx\in \bfZ^n_2}(-1)^{P_\psi(\bfx)}\ket{\bfx}.
\end{align}

When $k=2$, we call it a graph state~\cite{bell2014}. For convenience, if the specification of $k=2$ is not necessary, we shall call a hypergraph state even for graph states without losing its rigor. Here, $P_{\psi}$ denotes the corresponding $k$th-degree Boolean polynomial. Reuse of index $k$ from the Clifford hierarchy to hypergraph states is reasonable since $U_E\equiv\prod_{A\in E}C_A Z$ is of the $ k$th-order Clifford hierarchy. See App.~\ref{app:CCZ_hierarchy} for its proof.  Equation~\eqref{main:def_hypergraph_state} is the target state of our main interest. Moreover, we denote the orthonormal bases generated from the hypergraph state as $\left\{\ket{\psi_{\bfa}}\equiv Z^{\bfa}\ket{\psi}\right\}_{\bfa\in \bfZ^n_2}$ where $\ket{\psi_{\mathbf{0}}}:=\ket{\psi}$.

\section{Tailoring of states}

In practice, the hypergraph state is affected by noise, $\rho=\mathcal{N}(\ket{\psi}\bra{\psi})$, where $\mathcal{N}$ is a noise channel unknown due to environmental interactions. We assume the noise channel is the same for each input-state copy and time-independent.  We also refer to $\rho$ as $k$th-order state, following the order of $\ket{\psi}$. The objective of this work is to estimate the dephasing noise distribution $p:\bfZ^n_2\rightarrow \mathbb{R}_{\ge0}$ such that $p_{\bfa}=\braket{\psi_\bfa|\rho|\psi_\bfa}\;(\bfa\in \bfZ^n_2)$. We have two reasons for this. The first is that knowing the support size of $p$ is helpful because it sufficiently bounds the degree of noise dispersion~\cite{gross2010}. A detailed expression is followed by

\begin{proposition}\label{prop:dominant_rank}
    Given $\epsilon>0$ and $p_{\bfa}=\braket{\psi_{\bfa}|\rho|\psi_{\bfa}}\; (\bfa\in \bfZ^n_2)$, we suppose there exists a subset $A\subset \bfZ^n_2$ satisfying $\sum_{\bfa\in A}p_{\bfa}>1-\epsilon$. Then the spectral decomposition $\rho=\sum_{\bfa\in \bfZ^n_2}\lambda_{\bfa}\ket{\lambda_{\bfa}}\bra{\lambda_{\bfa}}$ also has subset $B\subset \bfZ^n_2$ such that $|B|=|A|$ and  $\sum_{\bfa\in B}\lambda_{\bfa}>1-\epsilon$.  
\end{proposition}

\begin{proof}
     We relocate the elements of $p$ in descending order, and then denote $A^{\downarrow}$ as the set of binary vector indices of the $|A|$ number of largest values in descending order. Then it is definite that $\sum_{\bfa\in A^{\downarrow}}p_{\bfa}\ge\sum_{\bfa\in A}p_{\bfa}>1-\epsilon$. Next, we see that $p_{\bfa}=\sum_{\bfb}|\braket{\psi_\bfa|\lambda_{\bfb}}|^2\lambda_{\bfb}$. Then we can rewrite it as $p=B\lambda$ where $B_{\bfa,\bfb}=|\braket{\psi_\bfa|\lambda_{\bfb}}|^2$ is a bistochastic matrix. Therefore, $\lambda$ majorizes $p$~\cite{arnold2012}. We set $B$ as the set of binary vector indices of the first $|A|$ descending values of $\lambda$. Then $|B|=|A|$, and we conclude that $\sum_{\bfa\in B}\lambda_\bfa\ge\sum_{\bfa\in A^{\downarrow}}p_{\bfa}>1-\epsilon$, where the first inequality holds by the majorization property. 
\end{proof}

Second, more importantly, we can transform $\rho$ into a dephased form; the off-diagonal elements are removed before being injected into the quantum simulation. The latter is concretized as the following statement, which generalizes the results of Ref.~\cite{piveteau2021error}, 



\begin{proposition}\label{main:dephasing_operation}
    (i) Given an ideal state $\ket{\phi}=U\ket{+}^{\otimes n}$  where the unitary $U$ is of up to $k$th-order Clifford hierarchy, suppose its noisy state $\rho$ can be accessed. Then there exist randomized operations of $(k-1)$th-order Clifford hierarhcy that transform $\rho$ into $\rho_{p}\equiv \sum_{\bfa\in \bfZ^n_2}p_{\bfa}\ket{\phi_\bfa}\bra{\phi_\bfa}$, where $\ket{\phi_\bfa}\equiv UZ^{\bfa}\ket{+}^{\otimes n}$ and $p$ is some probability distribution. Moreover, $p_{\bfa}$ is given by $p_{\bfa}=\braket{\phi_{\bfa}|\rho_p|\phi_{\bfa}}=\braket{\phi_{\bfa}|\rho|\phi_{\bfa}},~\forall \bfa\in \bfZ^n_2$.\\
    (ii) When $\phi$ is a stabilizer state and  undergoes Pauli noise, $\rho=\rho_p$ and such randomized operation is not needed.
\end{proposition}

\begin{proof}
    
We recall that ${\rm Cl}_n^{(k)}$ is denoted as a set of unitaries of $k$th-order Clifford hierarchy~\cite{anderson2024}. We note that ${\rm Cl}_n^{(1)}=\mathcal{P}_n$ and ${\rm Cl}_n^{(2)}={\rm Cl}_n$.

We first prove the noise tailoring part. Suppose the desired state is $\ket{\phi}=U\ket{+}$, where $U\in {\rm Cl}_n^{(k)}$. We set $X_{\phi}^{\bfa}\equiv UX^{\bfa}U^{\dag}$ ($\bfa\in \bfZ^n_2$) and $\mathcal{U}_t(\cdot)\equiv \frac{1}{2^n}\sum_{\bfa\in\bfZ^n_2}X_{\phi}^{\bfa}(\cdot)X_{\phi}^{\bfa}$. Note that $UX^{\bfa}U^{\dag}\in {\rm Cl}^{(k-1)}_n$~\cite{cui2017}, and it can be found within $\mathcal{O}(n^3)$-time. If $k=3\;(2\;\rm resp.)$, then it follows that $\mathcal{U}_t$ is a randomized Clifford (Pauli) operation. Next. we show that such channel tailors the noise of $\rho$. Recalling $\left\{\ket{\phi_{\bfa}}\right\}_{\bfa\in\bfZ^n_2}$ forms orthonormal bases, the output state is expressed by, 
    \begin{align}
        \mathcal{U}_t(\rho)&=\frac{1}{2^n}\sum_{\bfa,\bfb,\bfb'\in\bfZ^n_2}\braket{\phi_{\bfb}|\rho|\phi_{\bfb'}}X_{\phi}^{\bfa}\ket{\phi_{\bfb}}\bra{\phi_{\bfb'}}X_{\phi}^{\bfa}\nonumber\\&=\frac{1}{2^n}\sum_{\bfa,\bfb,\bfb'\in\bfZ^n_2}(-1)^{\bfa\cdot(\bfb+\bfb')}\braket{\phi_{\bfb}|\rho|\phi_{\bfb'}}\ket{\phi_{\bfb}}\bra{\phi_{\bfb'}}\nonumber\\&=\sum_{\bfb,\bfb'\in\bfZ^n_2}\delta_{\bfb,\bfb'}\braket{\phi_{\bfb}|\rho|\phi_{\bfb'}}\ket{\phi_{\bfb}}\bra{\phi_{\bfb'}}\nonumber\\&=\sum_{\bfb\in\bfZ^n_2}\braket{\phi_{\bfb}|\rho|\phi_{\bfb}}\ket{\phi_{\bfb}}\bra{\phi_{\bfb}}=\rho_{p}\label{eq:orthogonal_fidelity},
    \end{align}
    where the probability distribution $p$ is given as $p_{\bfa}=\braket{\phi_{\bfa}|\rho|\phi_{\bfa}}$ and second equality utilizes the fact that $UX^{\bfa}U^{\dag}\ket{\phi_{\bfb}}=UX^{\bfa}U^{\dag}UZ^{\bfb}\ket{+}^{\otimes n}=UX^{\bfa}Z^{\bfb}\ket{+}^{\otimes n}=(-1)^{\bfa\cdot \bfb}UZ^{\bfb}\ket{+}^{\otimes n}=(-1)^{\bfa\cdot \bfb}\ket{\phi_{\bfb}}$. Finally, we note that $\braket{\phi_{\bfa}|\rho_p|\phi_{\bfa}}=\braket{\phi_{\bfa}|\rho|\phi_{\bfa}}$ from the last of Eq.~\eqref{eq:orthogonal_fidelity}.

   If $k>3$, we need at most $\mathcal{O}(n^{k-1})$ number of $(k-1)$th-order multiple controlled $Z$ gate ($X$ gate resp.), each of which can be implemented via $\mathcal{O}(k)$-sized Clifford+T gates~\cite{nielsen2001}. Even though $E$ has mixed orders, the required number of gates is $\mathcal{O}(3n^2+4n^3+\cdots+kn^{k-1})=\mathcal{O}(kn^{k-1})$. We also note that this result can be applied to every $\ket{\phi}=U\ket{+}^{\otimes n}$ where $U\in {\rm Cl}_n^{(3)}$. 

   If $\ket{\phi}$ is a stabilizer state, then every Pauli noise gives a phase flip of stabilizer generators~\cite{nielsen2001,aaronson2004}. Therefore, all Pauli noise can be regarded as $UZ^{\bfa}$ operation for some $\bfa\in \bfZ^n_2$, and then $\rho=\rho_p$. The proof of Prop.~\ref{main:dephasing_operation} is completed.
\end{proof}

Here, $\ket{\phi}$ was used for general statements over hypergraph states $\ket{\psi}$. 
Overall, the proposition offers a method to diagonalize the input state via the gates from one level lower in the Clifford hierarchy, which is at least easier to implement than the original gate for the state generation. 

Since we mainly focus on controlling and learning the noise in hypergraph states using a quantum circuit that is relatively less noisy than the one used for state preparation, we henceforth assume the tailoring operations to be noiseless.
In particular, unitaries of the third-order Clifford hierarchy generate third-order hypergraph states. In this case, such randomized operations are $CZ$ gates with Pauli-$X$ operations. This implies that Clifford-based universal computing~\cite{litinski2019} with third-order state injection can accommodate the tailored $Z$-noise while preserving its overall Clifford structure. Compared to third-order ($CCZ$) gates, high-fidelity realizations of $CZ$ gates have been demonstrated much more extensively under a variety of techniques, both in superconducting qubits~\cite{Li2024,Li:2019realisation} and in Rydberg-atom platforms~\cite{evered2023}. Target graph states are tailored by a random Pauli operation only, and do not even require such an operation under the Pauli noise. For this case, knowing $\rho_p$ directly implies knowing $\rho$. 

Proposition~\ref{main:dephasing_operation} indicates that even if the input hypergraph state has unknown noise, we can use such twirling to encode it as a dephased hypergraph state, not harming the diagonal elements, referred to as noise tailoring.
This is analogous to the noise tailoring for the noisy channels~\cite{wallman2016,harper2020,liu2024,emerson2007}. However, Proposition~\ref{main:dephasing_operation} focuses on tailoring of noise in the state input, for which its preparation scheme can be general (i.e., from magic state distillation or flag-qubit method)~\cite{litinski2019,jhuang2024,bravyi2012,lee2024,chamberland2019}.
An additional difference from an operational point of view is that for noisy states, we only need to apply the twirling operation after the noisy channel. In contrast, in channel twirling, the operation appears both before and after the target noise channel~\cite{liu2024,tsubouchi2024}. A recent technique~\cite{lee2025_benchmarking} offers a method to make the target noisy magic state block-diagonalized. Meanwhile, our method guarantees complete diagonalization using commuting gate resources. Noise tailoring leads to a significant reduction of target noise complexity, and we shall see that this will lead to a systematic algorithm to learn the noise of hypergraph states.

Now, let us assume the input state $\rho$ is tailored into $\rho_p$. One of the best ways for estimating the distribution $p$ of noisy hypergraph state $\rho_p$ is to measure $\rho$ with respect to orthonormal bases $\left\{\ket{\psi_\bfa}\right\}_{\bfa\in \bfZ^n_2}$. 
However, it requires $\mathcal{O}(n^k)$ number of two-qubit gate resources, similarly to prepare $\ket{\psi}$. An alternative way is to estimate $p_{\bfa}$ by the shadow tomography~\cite{aronson2018,huang2020,koh2022,bertoni2024,zhang2023,hu2023,wang2024}.  
However, this protocol requires estimating $p_{\bfa}$ for all $\bfa\in \bfZ^n_2$, and calculating classical shadow~\cite{huang2020,bertoni2024} becomes inefficient when we have to learn the noise of a multi-qubit entangled magic state.   
\begin{figure}[t]
    \includegraphics[width=\columnwidth]{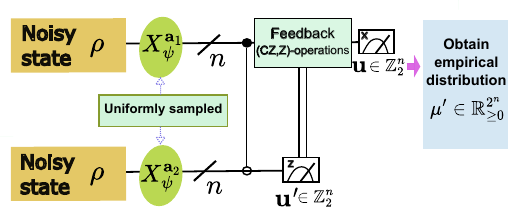}
    \caption{Schematic diagram of noise tailoring and learning of noisy 3rd-order hypergraph state $\rho$. Two qubit gates of $X^{\bfa_1\;(\bfa_2)}_\psi$ can be merged into the feedback section.}
    \label{fig:detection diagram}
\end{figure}

\section{Noise convolution}

To address these problems in the conventional techniques presented in the previous section, we propose an efficient noise learning scheme specialized to hypergraph states. As we shall see, the sampling complexity and time complexity are significantly reduced while achieving similar (or even better)
accuracy. As an initial step, we show in Theorem~\ref{thm:main_2markovian_eq} that by measuring two copies of $\rho_{p}$ with the unitaries of one-step lower Clifford hierarchy and transversal \textsc{CNOT} gates applied to the state, we can obtain a quadratic equation of $p$. 

\begin{theorem} [Noise convolution for graph states]\label{thm:main_2markovian_eq} Suppose that $\rho$ is a noisy $k(\in \mathbb{N})$th-order hypergraph state of $n$ qubits and $p$ is the diagonal elements such that $p_{\bfa}=\braket{\psi_\bfa|\rho|\psi_{\bfa}}$. Then, we can obtain the measurement outcomes $\bfu\in \bfZ^n_2$ following the probability distribution $(\mu_{\bfu})_{\bfu\in \bfZ^n_2}$ satisfying
    \begin{align}\label{eq:thm_main_2markovian_eq}
        (p\ast p)_{\bfu}\equiv \sum_{\bfa\in \bfZ^n_2}p_{\bfa}p_{\bfa+\bfu}=\mu_{\bfu},
\end{align} using operations of $(k-1)$th-order Clifford hierarchy, transversal \textsc{CNOT}, and Pauli measurements.
\end{theorem}

Equation~\eqref{eq:thm_main_2markovian_eq} is the core equation throughout this paper, which we call the convoluted noise equation. This term naturally comes from the fact that the left side $p\ast p$ is the self-convolution (or auto-correlation)~\cite{bu2023} of $p$. The nonlinearity of such an equation can be interpreted as a trade-off when we use one-step lower Clifford hierarchy resources. When $\ket{\psi}$ is a graph state ($k=2$), the process requires only a single-depth transversal \textsc{CNOT}s, which has a significant reduction over $\mathcal{O}(n)$-depth syndrome measurement-based noise learning~\cite{Wagner2023,nielsen2001}.

We outline the quantum measurement algorithm for obtaining the empirical distribution of $\mu$, which schematically summarizes the proof of Thm.~\ref{thm:main_2markovian_eq}. A detailed proof will be shown at the end of this section. We define the directional derivative of $P_\psi(\bfx)$ as the Boolean polynomial $P_{\psi}(\bfx|\bfx')$ as follows $P_\psi(\bfx+\bfx')+P_{\psi}(\bfx)$. Then the scheme is shown as follows: In a $2n$-qubit system, we prepare two copies of $\rho$, $\rho \otimes \rho$.  Then we uniformly sample $\bfa_1,\bfa_2\in \bfZ^n_2$, enact $X_{\psi}^{\bfa_1\;(\bfa_2\;{\rm resp.})}\equiv U^{1(2)}_EX^{\bfa_1(\bfa_2)}U_E^{1(2)\dagger}$ to the first (second) component ($U_E\equiv \prod_{A\in E}C_{A}Z$), and then $\prod_{i=1}^{n}{\textsc{CNOT}}_{i,i+n}$ follows. We measure the $n+1\sim2n$-th qubits in the Pauli $Z$-basis and obtain the measurement outcome $\bfu'\in \bfZ^n_2$. As feedback, we enact another multiple controlled-$Z$ gates $V_{\psi,\bfu'}$ such that $V_{\psi,\bfu'}\ket{\bfx}=(-1)^{P_{\psi}(\bfx|\bfu')}\ket{\bfx}$, to the remaining qubits. The output resembles the two successive operations of $U_E$~\cite{catani2018} with a dephasing noise, where two $U_E$ commute with such noise and vanish, and hence becomes $\ket{+}^{\otimes n}$ followed by two successive dephasing noises $p$. Finally, $X$-basis measurement outcome $\bfu$ follows the distribution of $ \mu_{\bfu}=\sum_{\bfa\in \bfZ^n_2}p_{\bfa}p_{\bfa+\bfu}$, which is the noise rate of successive dephasing noises. 

An important remark is that the multiple controlled $Z$ gates of $X^{\bfa_1\;(\bfa_2)}_\psi$ can be merged into $V_{G,\bfu'}$ to become $V_{G,\bfu'}^{\bfa_1,\bfa_2}$, because the control gates of $X^{\bfa_1\;(\bfa_2)}_\psi$ merely yields additional \textsc{CZ}-gates while commuting with the transversal \textsc{CNOT} part. Thus, we only need to give multiple controlled-Z gate operations before the final measurement.
The validity of such a merging technique is presented in App.~\ref{app:gate_merging}. 

Figure~\ref{fig:detection diagram} diagrammatically illustrates the noise tailoring and sampling under the distribution $\mu$ for the third-order hypergraph state cases. In App.~\ref{app:gate_merging}, we showed that in this case, our algorithm only uses Clifford resources and hence does not exceed on average the $\Theta(\frac{n^2}{4})$ number (at most $\mathcal{O}(\frac{n^2}{2})$)  of two-qubit Clifford gate requirements without any magic gates, while direct measurement with the bases $\left\{\ket{\psi_\bfa}\bra{\psi_\bfa}\right\}_{\bfa\in \bfZ^n_2}$ requires $\mathcal{O}(n^3)$ Clifford+T resources~\cite{bravyi2016i}. Furthermore, most practical hypergraph states to achieve universality~\cite{miller2016}, such as the Union Jack state~\cite{zhu2019},  can be implemented on 2D-local connectivity~\cite{jhuang2024}, in which case only constant-depth \textsc{CZ}-gates are needed. For the graph state case, single-depth transversal \textsc{CNOT} gates are the only gate resources. The reason is that $X^{\bfa_1\;(\bfa_2)}_{\psi}$ is a Pauli operator which can be regarded as classical post-processing of measurement outcome. 

We end this section with a rigorous proof of Thm.~\ref{thm:main_2markovian_eq}.

\begin{proof}[Proof of Thm.~\ref{thm:main_2markovian_eq}]
     Let $\rho=\mathcal{N}(\ket{G(V,E)}\bra{G(V,E)})=\mathcal{N}(\ket{\psi}\bra{\psi})$ be an $n$-qubit hypergraph state having unknown noise $\mathcal{N}$. We also define $U=\prod_{A\in E}C_{A}Z$. By   Prop.~\ref{main:dephasing_operation}, we act the tailoring operation to evolve $\rho$ to $\rho_{p}$ with a given probability distribution $p$ over $\bfZ^n_2$. Then we start from this input state. This is a reasonable assumption because all operations in our algorithm are linear. Note that in this case $\ket{\psi_{\bfa}}=Z^{\bfa}\ket{\psi}$, since $U$ and Pauli $Z$ operations commute. Now, we append another $\rho_{p}$, and enact $\prod_{i=1}^{n}{\textsc{CNOT}}_{i,i+n}$, where ${\textsc{CNOT}}_{i,j}\;(i,j\in [n])$ refers to a \textsc{CNOT} gate from control $i$-th qubit to target $j$-th qubit.
The evolved state is expressed by, 
\begin{align}\label{eq:gadgetization}
    &\prod_{i=1}^{n}{\textsc{CNOT}}_{i,i+n}\left(\rho_p\otimes \rho_p\right)\prod_{i=1}^{n}{\textsc{CNOT}}_{i,i+n}\nonumber
\end{align}
\begin{align}    
    &=\frac{1}{2^{2n}}
    \sum_{\substack{\bfx,\bfx',\bfy,\bfy'\in \bfZ^n_2\\\bfa,\bfb\in\bfZ^n_2}}p_{\bfa}p_{\bfb}(-1)^{\bfa\cdot(\bfx+\bfy)+\bfb\cdot(\bfx'+\bfy')}\nonumber\\
  	&\cdot(-1)^{P_{\psi}(\bfx)+P_{\psi}(\bfy)+P_{\psi}(\bfx')+P_{\psi}(\bfy')}\ket{\bfx(\bfx'+\bfx)}\bra{\bfy(\bfy'+\bfy)}\nonumber\\
  	&= \frac{1}{2^{2n}}
    \sum_{\substack{\bfx,\bfx',\bfy,\bfy'\in \bfZ^n_2\\\bfa,\bfb\in\bfZ^n_2}}p_{\bfa}p_{\bfb}(-1)^{\bfa\cdot(\bfx+\bfx'+\bfy+\bfy')+\bfb\cdot(\bfx'+\bfx+\bfy'+\bfy)}\nonumber\\&\cdot(-1)^{P_{\psi}(\bfx)+P_{\psi}(\bfy)+P_{\psi}(\bfx+\bfx')+P_{\psi}(\bfy+\bfy')}\ket{\bfx\bfx'}\bra{\bfy\bfy'}\nonumber
\end{align}
\begin{align}
    &=\frac{1}{2^{2n}}
    \sum_{\substack{\bfx,\bfx',\bfy,\bfy'\in \bfZ^n_2\\\bfa,\bfb\in\bfZ^n_2}}p_{\bfa}p_{\bfb}(-1)^{\bfa\cdot(\bfx+\bfx'+\bfy+\bfy')+\bfb\cdot(\bfx'+\bfx+\bfy'+\bfy)}\nonumber\\&\cdot (-1)^{P_{\psi}(\bfx|\bfx')+P_{\psi}(\bfy|\bfy')}\ket{\bfx\bfx'}\bra{\bfy\bfy'},
\end{align}
where $P_{\psi}(\bfx|\bfx')\equiv P_\psi(\bfx)+P_\psi(\bfx+\bfx')$ is a binary polynomial which has $(k-1)$th-degree for $\bfx$ variables. Because both $P_\psi(\bfx)$ and $P_{\psi}(\bfx+\bfx')$ have the same $k$th-degree monomials with $\bfx$, and hence they vanish after the summation. Next, we measure from $n+1$ to $2n$-th qubits in the Pauli-$Z$ bases, and suppose we got an outcome $\bfu'\in\bfZ^n_2$. The probability to get the outcome $\bfu'$, say $\rm Prob(\bfu')$, is as follows, 
{\small
\begin{align}\label{eq:prob_sub_outcome}  
    &{\rm Prob}(\bfu')\nonumber\\&=\tr{\prod_{i=1}^{n}{\textsc{CNOT}}_{i,i+n}\left(\rho_p\otimes \rho_p\right)\prod_{i=1}^{n}{\textsc{CNOT}}_{i,i+n}(I\otimes \ket{\bfu'}\bra{\bfu'})}\nonumber
\end{align}
\begin{align}
    &=\frac{1}{2^{2n}}
    \sum_{\substack{\bfx,\bfy\in \bfZ^n_2\\\bfa,\bfb\in\bfZ^n_2}}p_{\bfa}p_{\bfb}(-1)^{\bfa\cdot(\bfx+\bfy)+\bfb\cdot(\bfx+\bfy)+P_{\psi}(\bfx|\bfu')+P_{\psi}(\bfy|\bfu')}\delta_{\bfx,\bfy}\nonumber\\&
    =\frac{1}{2^{2n}} \sum_{\bfx,\bfa,\bfb\in \bfZ^n_2}p_{\bfa}p_{\bfb}=\frac{1}{2^n}, \forall \bfu'\in\bfZ^n_2.
\end{align} 
}

A remark is that when we run the scheme in a real situation, $\bfu'$ may not be sampled uniformly. Because we randomly choose $X^{\bfa}_\psi$ only once before each measurement. Hence, the input state to be measured is not $\rho_p$. However, the above reasoning is safe because we only consider the analytic output after the randomization. In other words, all processed operations are linear, and then we may regard the output to be randomized over $\forall \bfa\in \bfZ_2^n$.

Let us get back on track. The projected state  $\rho_{p,\bfu'}$ over the remaining qubits becomes,
{\small
\begin{align}
    &\rho_{p,\bfu'}\nonumber\\&=\frac{1}{2^{n}}
    \sum_{\substack{\bfx,\bfy\in \bfZ^n_2\\\bfa,\bfb\in\bfZ^n_2}}p_{\bfa}p_{\bfb}(-1)^{(\bfa+\bfb)\cdot(\bfx+\bfy)+P_{\psi}(\bfx|\bfu')+P_{\psi}(\bfy|\bfu')}\ket{\bfx}\bra{\bfy}.
\end{align}
}


Then we can find up to $\mathcal{O}(\comb{n}{k-1})$ number of $(k-1)$th-order \textsc{CZ} gates, or $\mathcal{O}(kn^{k-1})$ number of Clifford+T gates~\cite{nielsen2001} to post-process $V_{\psi,\bfu'}$ such that $V_{\psi,\bfu'}\ket{\bfx}\bra{\bfy}V_{\psi,\bfu'}= (-1)^{P_{\psi}(\bfx|\bfu')+P_{\psi}(\bfy|\bfu')}\ket{\bfx}\bra{\bfy}$ for $\forall \bfx,\bfy\in \bfZ^n_2$. If $k=3$, then only Clifford gates are necessary. If $k=2$, even \textsc{CZ} gates are not needed. After that, we finally obtain the evolved state,
\begin{align}
    \rho'_p\equiv \frac{1}{2^{n}}
    \sum_{\substack{\bfx,\bfy\in \bfZ^n_2\\\bfa,\bfb\in\bfZ^n_2}}p_{\bfa}p_{\bfb}(-1)^{(\bfa+\bfb)\cdot(\bfx+\bfy)}\ket{\bfx}\bra{\bfy}.
\end{align}
We note that this is a common output state for $\forall \bfu'\in\bfZ^n_2$, i.e. deterministic. Now, we measure the remaining qubits with a Pauli-$X$ basis. Suppose we obtain the final outcome as $\ket{+^{\bfu}}\equiv Z^{\bfu}\ket{+}^{\otimes n}$. Then its measurement probability $\mu_{\bfu}\equiv\braket{+^\bfu|\rho'_p|+^\bfu}$ is,
\begin{align}\label{eq:prob_equation}
    \mu_{\bfu}&=\frac{1}{4^n}\sum_{\substack{\bfx,\bfy\in \bfZ^n_2\\\bfa,\bfb\in\bfZ^n_2}}p_{\bfa}p_{\bfb}(-1)^{(\bfa+\bfb+\bfu)\cdot(\bfx+\bfy)}\nonumber\\&=\sum_{\bfa,\bfb\in\bfZ^n_2}p_{\bfa}p_{\bfb}\delta_{\bfa+\bfu,\bfb}=\sum_{\bfa\in\bfZ^n_2}p_{\bfa}p_{\bfa+\bfu}.
\end{align}
The proof is completed.
\end{proof}

\section{Decoding the self-convolution}

After the algorithm outlined in Theorem~\ref {thm:main_2markovian_eq}, a key question is; what is the method for decoding the self-convolution in Eq.~\eqref{eq:thm_main_2markovian_eq}, to estimate the noise rates within a certain error $\epsilon$? We divide the decoding methods into two categories: Walsh-Hadamard transform (WHT) and convolution power method (CPM). 

\subsection{Walsh-Hadamard transform (WHT)}

We should sample the measurement outcomes following $\mu$ in Eq.~\eqref{eq:thm_main_2markovian_eq} sufficiently so that the empirical distribution of $\mu'$ becomes close to $\mu$. 
After we substitute $\mu'$ for $\mu$, we can solve the equation to find the solution $p'$ by taking the Walsh-Hadamard transform (WHT)~\cite{alman2023,scheibler2015} on both sides. To be specific, for $\bfb\in \bfZ^n_2$, we find that
\begin{align}\label{eq:process_FWHT}
    \hat{\mu}_{\bfb}\equiv \sum_{\bfu\in \bfZ^n_2}(-1)^{\bfb\cdot\bfu}\mu_{\bfu}&=\sum_{\bfu\in \bfZ^n_2}\sum_{\bfa\in \bfZ^n_2}(-1)^{\bfb\cdot \bfu}p_{\bfa}p_{\bfa+\bfu}\nonumber\\&=\sum_{\bfa,\bfc\in \bfZ^n_2}(-1)^{\bfb\cdot(\bfa+\bfc)}p_{\bfa}p_{\bfc}\nonumber\\&=\left(\sum_{\bfa\in \bfZ^n_2}(-1)^{\bfa\cdot \bfb}p_{\bfa}\right)^2\nonumber\\&=(\hat{p}_{\bfb})^2
\end{align}
where $\cdot$ means the binary inner product. By using the Hadamard matrix $H=(H_{\bfa,\bfb})_{\bfa,\bfb\in \bfZ^n_2}$ that is defined as $H_{\bfa,\bfb}=(-1)^{\bfa\cdot \bfb}$, we obtain the following result,
\begin{align}\label{eq:main_exactsolution_fwht}
    p'=\frac{1}{2^n}H\sqrt{H\mu'}=\widehat{\sqrt{\widehat{\mu'}}}^{-1}\; (\widehat{\mu}\equiv H\mu).
\end{align}
Here, we also used the fact that $H^{-1}=\frac{1}{2^n}H$.
The hat notation is used for WHT, and the square root means we take the square root of each element in the input. 

Since we selected the positive root, Eq.~\eqref{eq:main_exactsolution_fwht} is well-defined only when all elements in $\widehat{p}=Hp$ are non-negative. Such a WH spectrum typically resides between $(1,1,1\ldots,1)$ of the pure case, to $(1,0,0,\ldots,0)$ of the totally mixed case. The reason is that if we define the matrix function $P(\tau)$ such that $P(\tau)_{\bfa,\bfb}\equiv \tau_{\bfa+\bfb}$, then $\widehat{\tau}$ becomes eigenvalue spectrum of $P$ with the eigenvectors $\eta^{(\bfc)}$ such that $\eta^{(c)}_{\bfa}\equiv (-1)^{\bfc\cdot\bfa}$.  Given that $\delta\le\frac{1}{2}$, $\forall\widehat{p}_{\bfa}\ge p_{\mathbf{0}}-(1-p_{\mathbf{0}})=1-2\delta\ge 0$. We note that $P(\mu)=P^2(p)\ge 0$ and then $\hat{\mu}$ is always non-negative. Additionally, we note that $p'$ is quasi-probability~\cite{mari2012}. In other words, it sums to $1$ since 
\begin{align}
    \sum_{\bfb\in \bfZ_2^n}p_{\bfb}'=(\sqrt{H\mu'})_{\mathbf{0}}=\sqrt{(H\mu')_\mathbf{0}}=\sqrt{\sum_{\bfb}\mu'_\bfb}=1
\end{align}
but could allow negative elements. In the later section, we will show how to efficiently peel off such negative elements while improving the accuracy with the target distribution.

In general, the fast Walsh Hadamard transform (FWHT) algorithm enables us to perform the calculation of $p$ within $\mathcal{O}(n2^n)$~\cite{alman2023} time which can be reduced to $\mathcal{O}(|V|\log (|V|))$ if the support of $\mu'$ is known to belong to proper subspace $V\subset\bfZ^n_2$ (see App.~\ref{app_time_solving_noise_equation} for the proof).
Theoretically, we hope the solution $p'$ to be close enough to the true solution $p$ when $\mu'$ is close to the true measurement distribution $\mu$. Indeed, the following statement fulfills such requirements.
\begin{proposition}\label{main:thm_uniqueness}
    
    Suppose $p^{(1)},p^{(2)}$ be solutions in Eq.~\eqref{eq:main_exactsolution_fwht} ,corresponding to $\mu^{(1)},\mu^{(2)}$ respectively. 
    \\
    (i) We assume that $1-p^{(1)}_{\mathbf{0}},1-p^{(2)}_{\mathbf{0}}<\delta\ll\frac{1}{2}$. There exists $\epsilon_{\rm th}>0$ such that $\|\mu^{(1)}-\mu^{(2)}\|<\epsilon_{\rm th}$ implies we can obtain asymptotic bound of $\|p^{(1)}-p^{(2)}\|$ as
    \begin{align}\label{eq:asymptotic_p_dist}
     \frac{\|p^{(1)}-p^{(2)}\|_1}{\|\mu^{(1)}-\mu^{(2)}\|_1}&\lesssim\left(\frac{1}{2}+\delta\right)+\frac{2+9\delta}{8}\|\mu^{(1)}-\mu^{(2)}\|_1.
 \end{align}
 \\
    (ii) We assume that $1-p^{(1)}_{\mathbf{0}},1-p^{(2)}_{\mathbf{0}}<\delta<\frac{1}{2}-\Omega(1)$. Then,
    \begin{align}
        \frac{\|p^{(1)}-p^{(2)}\|_1}{\|\mu^{(1)}-\mu^{(2)}\|_1}=\mathcal{O}(2^n),\; {\rm and}\; \frac{\|p^{(1)}-p^{(2)}\|_2}{\|\mu^{(1)}-\mu^{(2)}\|_2}=\mathcal{O}(1).
    \end{align}
    \\
    (iii) $\|p-p'\|_2\le \mathcal{O}(\min_{\bfa}\widehat{\mu}_\bfa^{-1})\|\mu-\mu'\|_2$.

\end{proposition}
Here, $l_1$ ($l_2$)-norm is defined as $\|\bfv\in \mathbb{R}^{2^n}\|_{1\;(2\;{\rm respectively})} \equiv \sum_{\bfa\in \bfZ^n_2}|v_{\bfa}|$ $\left(\sqrt{\sum_{\bfa\in \bfZ^n_2}v_{\bfa}^2}\right)$. The proof and reasoning that $\epsilon_{\rm th}$ is independent of $n$ are deferred to App.~\ref{app_time_solving_noise_equation}. A naive calculation of Eq.~\eqref{eq:main_exactsolution_fwht} takes $\mathcal{O}(n2^n)$-time~\cite{scheibler2015} and thus may be inefficient. In the next subsection, we propose an alternative method, the convolution power method (CPM), to address this issue.

\subsection{Convolution power method (CPM)}\label{main_sec_CPM}

Given small noise, we can expand the solution with the convolution powers $(\mu\ast \mu\cdots\ast \mu)$~\cite{luchko2022} and sample the outcomes from each term to obtain the approximate solution from $p_{\rm approx}=\sum_{j=1}^{\infty}c_j\mu^{\ast j}$. Here, $\forall c_j\in \mathbb{R}$, and $\mu^{\ast j}\;(k\in \mathbb{N})$ means the $j$th-multiple convolutions of $\mu$ ($\mu^{\ast 0}\equiv (1,0,0,\ldots,0)=:\vec{1},\;\mu^{\ast 1}=\mu,\; \mu^{\ast 2}=\mu\ast \mu,\ldots$). Here, we briefly outline the analytical process used to obtain the approximate solution, while detailed explanations are provided in App.~\ref{method:approx_solution}. Then, we present the polynomial-complexity arguments for CPM in the next section. 

We define $\bm{\delta}_p\equiv \vec{1}-p$. Then the convoluted noise equation becomes $(\vec{1}-\bm{\delta}_p)\ast p=\mu$. We take the inverse convolution $p^{\ast -1}$ satisfying $p^{\ast -1}\ast p=\vec{1}$. The convolution forms a commutative and associative algebra, where $\vec{1}$ is the unity. Therefore, we obtain a natural geometric series formula, $(\vec{1}-\bm{\delta}_p)^{\ast-1}=\sum_{l=0}^{\infty}\bm {\delta}_p^{\ast l}$. After that, the equation turns into 
{\small
\begin{align}
    p=(\vec{1}-\bm{\delta}_p)^{\ast -1}\mu&=\sum_{l=0}^{\infty}\bm{\delta}_p^{\ast l}\ast \mu=\sum_{l=0}^{w-1}\bm{\delta}_p^{\ast l}\ast \mu+\mathcal{O}((2\delta)^w),
\end{align}
}
with the predetermined approximation order $w\in \mathbb{N}$.
The rightmost term can be expressed with only the convolution between convolution powers of $p=\vec{1}-\bm{\delta}_p$ and $\mu$. The bias $\mathcal{O}((2\delta)^w)$ is measured as the $l_1$-norm ($\ge$ $l_2$-norm) of the vector difference. 

Next, we again express $\bm{\delta}_p$ as $\vec{1}-p$ to write the approximation as the summation of iterative convolutions of $p$ and $\mu$.  We note that all these convolution terms reduce to the convolution between single $p^{\ast (0\;\rm or \;1)}$ and convolution powers of $\mu$, by using the fact that $p\ast p=\mu$ (e.g. $p^{\ast 4}=\mu^{\ast2}$ and $p^{\ast 3}\ast \mu=p\ast \mu^{\ast 2}$). Hence, we obtain the linear equation, 
\begin{align}\label{eq:main_cpm_expansion}
    &\vec{1}\ast p+\left(\sum_{j=1}^{\lfloor\frac{w-1}{2}\rfloor+1}\bar{c}_j\mu^{\ast j}\right)\ast p=\sum_{j=1}^{\lfloor\frac{w-1}{2}\rfloor+1}c''_j\mu^{\ast j}\nonumber\\&+\mathcal{O}((2\delta)^w)\;(\forall \bar{c}_j,c''_j\in \mathbb{R})\nonumber\\&\Rightarrow\;p=\left(\vec{1}+\sum_{j=1}^{\lfloor\frac{w-1}{2}\rfloor+1}c'_j\bm{\delta}_{\mu}^{\ast j}\right)^{-1}\ast \frac{\sum_{j=1}^{\lfloor\frac{w-1}{2}\rfloor+1}c''_j\mu^{\ast j}}{1+\sum_{j=1}^{\lfloor\frac{w-1}{2}\rfloor+1}\bar{c}_j}\nonumber\\&+\mathcal{O}(\delta^w)\;(\forall c'_j\in \mathbb{R}),
\end{align}
where the reduction of the bias scaling is because we can show that both sides of the last fraction are $\Theta(2^{w})$. 

Now, it seems that we need another approximation order $s\in \mathbb{N}$ to truncate the geometric expansion of the inverse term in Eq.~\eqref{eq:main_cpm_expansion} by the convolution powers of $\sum_{j=1}^{\lfloor\frac{w-1}{2}\rfloor+1}c'_j\bm{\delta}_{\mu}^{\ast j}$, up to $(w+s)$ terms. We can also prove that the $l_{1\;({\rm or}\;2)}$-norm of $\sum_{j=1}^{\lfloor\frac{w-1}{2}\rfloor}c'_j\bm{\delta}_{\mu}^{\ast j}$ is lower or equal to $\frac{3w\delta}{2}\;(w\le 100)$. It means that given that $\delta<\frac{1}{3w}$, we can approximate the inverse with a small bias with size $\mathcal{O}\left(\left(\frac{3w\delta}{2}\right)^{w+s}\right)$. A more detailed explanation, including the bias scaling and the structure of both $\bar{c}_j,c'_j,$ and $c''_j$ can be found in App.~\ref{method:approx_solution}.

Finally, the truncation leaves $a(w,s)$ number of terms. We show that $a(w,s)=  \lfloor\frac{w+1}{2}\rfloor+(\lfloor\frac{w-1}{2}\rfloor+1)(w+s-1)\le (w^2+ws+w)$. As a result, $p$ is written by, 
{\small
\begin{align}\label{eq:after_inversion}
    p&=\sum_{j=1}^{a(w,s)}c^{(w,s)}_j \mu^{\ast j}+\mathcal{O}\left(\left(\frac{3w\delta}{2}\right)^{w+s}+\delta^w\right)\;(\forall c^{(w,s)}_j \in \mathbb{R}),
\end{align}
}given that $\delta<\frac{1}{3w}$. Such a noise-rate restriction is a rough bound. Indeed, we found that $\delta<\frac{1}{4}$ for $(w,s)=(2,0)$, and $\delta<\frac{1}{6}$ for $(w,s)=(3,0)$. Even after truncation, this approximate solution becomes exact as $\delta\rightarrow 0$. This means that the sum of the coefficients $c^{(w,s)}_j$'s of the approximate solution is $1$.

Let us give some remarks. First, we see that the approximate solution is a linear combination of convolution powers of $\mu$. It means that we need to estimate the convolution powers correctly so that the linear combination of these estimates has a desired small error. Such estimations are possible since we can sample $\bfu\in \bfZ^n_2$ following the $\mu^{\ast j} \;(j\in \mathbb{N})$. This is done by setting the sample $\bfu$ as $\bfa_1+\bfa_2+\cdots+\bfa_j+\bfa_{j+1}$, where these $\bfa's$ are sampled independently from $\mu$ and also used for estimating other $\mu^{\ast j'(<j)}$'s. Given $N$ number of copies from $\mu^{\ast 1},\mu^{\ast 2},\mu^{\ast3},\ldots$ respectively, we can obtain empirical distribution $\mu^{'\ast 1},\mu^{'\ast 2},\mu^{'\ast3},\ldots$ so that the final approximate solution $p'_{\rm approx}(=: p')$ is,
\begin{align}\label{eq:method_approx_sol}
    p'=\sum_{j=1}^{a(w,s)}c^{(w,s)}_j \mu^{'\ast j}\xrightarrow[N\rightarrow \infty]{} p^{(w,s)}_{\rm approx}=\sum_{j=1}^{a(w,s)}c^{(w,s)}_j \mu^{\ast j},
\end{align}
where each $c^{(w,s)}_j\in \mathbb{R}$ is found from App.~\ref{method:approx_solution}.

For example, when $s=0$, $w=2,3$ yields low-degree approximations, 
\begin{align}\label{eq:approx_solution_proto}
    &p^{(2,0)}_{\rm approx}=\frac{3}{2}\mu^{\ast 1}-\frac{1}{2}\mu^{\ast 2},\nonumber\\&p^{(3,0)}_{\rm approx}=\frac{111}{64}\mu^{\ast 1}-\frac{53}{64}\mu^{\ast 2}-\frac{3}{64}\mu^{\ast 3}+\frac{9}{64}\mu^{\ast 4}.
\end{align}
When $s=1$,
\begin{align}\label{eq:approx_solution_proto2}
     &p^{(2,1)}_{\rm approx}=\frac{7}{4}\mu^{\ast 1}-\mu^{\ast 2}+\frac{1}{4}\mu^{\ast 3},\nonumber\\&p^{(3,1)}_{\rm approx}\simeq 2.05\mu^{\ast 1}-1.672\mu^{\ast 2}+0.586\mu^{\ast 3}+0.141\mu^{\ast 4}\nonumber\\&-0.105\mu^{\ast 5}.
\end{align}

Second, we recall that the coefficients $c^{(w,s)}_j$'s sum to $1$. Since each $\mu^{'\ast j}$ is again a probability distribution, $p'_{\rm approx}$ is quasi-probabilistic: it sums to $1$, but allows some negative elements.

\begin{figure}[t]
    
    \centering

    \includegraphics[width=\columnwidth]{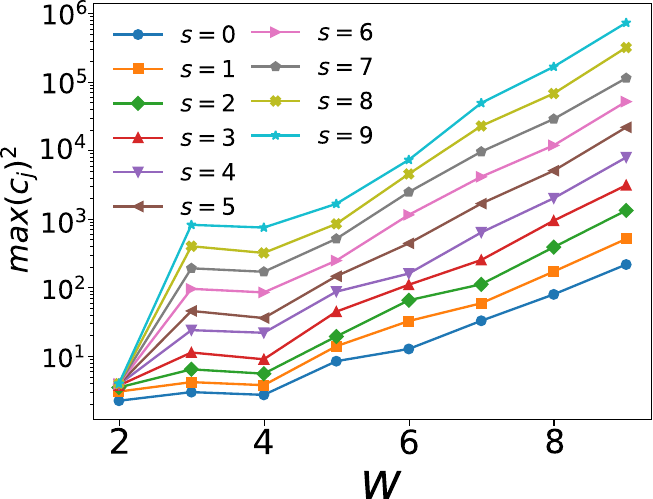}

    \caption{Scaling of the $\max\left\{|c^{(w,s)}_j|\right\}^2$ in Eq.~\eqref{eq:method_approx_sol} as we increase the degree of approximation $(w,s)$. }
    
\label{fig:c_j_degree_approx}
    
\end{figure}
\subsection{Efficient removal of negativity}\label{sec:negativity_removal}

In previous subsections, we confirmed that the solutions in WHT and CPM are quasi-probability distributions,
which sum to $1$ but allow negativity. For notational brevity, we shall denote the solution $p'_{\rm approx}=p'$ only for this subsection. After we obtain such a solution $p'$, we need to transform it into a genuine probability solution, say $p^{(+)}$. There is a well-known method to project the quasi-probability to the probability simplex while even reducing the $l_2$-distance (i.e., $\|p-p^{(+)}\|_2\le \|p-p'\|_2$) to the target distribution~\cite{wang2013,flammia2020,borwein2006}. We can also check that if the support size of $p'$ is denoted as $|A|$, then $p^{(+)}$ can be found in $\mathcal{O}(|A|\log(|A|)$~\cite{wang2013} time and memory (see App.~\ref{app:projecting_quasi_simplex_l2} for the detail). Moreover, a simpler method resolves the $l_1$ case~\cite{villani2009}, where the precise statement is given below:
\begin{proposition}
    Assume there is a quasi-probability $p'$ after the estimation of target distribution $p$ with the support size $|A|$ ($p'_{\bfa}\ne 0$ only when $\bfa\in A$). Then we obtain a transformation $p'\rightarrow p^{(+)}$ such that $p^{(+)}$ is a non-negative distribution and 
    \begin{align}
        \|p-p^{(+)}\|_1\le \|p-p'\|_1.
    \end{align}
    Such a transformation takes $\mathcal{O}(|A|)$-time and memory. 
\end{proposition}

\begin{proof}
    We denote $S=\sum_{\bfa\in \bfZ_2^n;p'_{\bfa}<0}|p'_\bfa|$. Then $\sum_{\bfa\in \bfZ_2^n;p'_{\bfa}\ge0}p'_{\bfa}=1+S$. We transform $p'$ into $p^{(+)}$ as follows: $p^{(+)}=\frac{p'_\bfa}{1+S}$ if $p'_{\bfa}\ge 0$, and $p^{(+)}=0$ otherwise. Then $p^{(+)}$ becomes a non-negative distribution. Moreover, 
    \begin{align}
        \|p-p^{(+)}\|_1&=\sum_{p'_{\bfa}\ge 0}|p_\bfa-\frac{1}{1+S}p'_\bfa|+\sum_{p'_{\bfa}<0}p_\bfa\nonumber\\&\le \sum_{p'_{\bfa}\ge 0}|p_\bfa-p'_\bfa|+\sum_{p'_\bfa<0}p_{\bfa}+\frac{S}{1+S}\sum_{p'_\bfa\ge0}p'_\bfa\nonumber\\&\le \sum_{p'_{\bfa}\ge 0}|p_\bfa-p'_\bfa|+\sum_{p'_\bfa<0}(p_\bfa-p'_\bfa)\nonumber\\&=\sum_{\bfa\in \bfZ_2^n}|p_\bfa-p'_\bfa|=\|p-p'\|_1.
   \end{align}
   Corresponding time and memory complexities of the transformation is definite. 
\end{proof}
Therefore, from now on, let us focus on obtaining $p'$ within $l_1$ or $l_2$ distance. This is because transforming $p'$ to $p^{(+)}$ ensures that the distance from the target distribution $p$ does not increase.

\section{Complexity of CPM}

Now, we analyze the sampling and time complexity of getting the approximate solution of CPM. We recall from Eq.~\eqref{eq:method_approx_sol} that the approximate solution is a linear combination with multiple convolutions of $\mu$. The next proposition indicates the required sampling copy number to achieve the desired estimation accuracy of each convolution power of $\mu$: 

\begin{proposition}\label{main:prop_sparsity_general}
    Given a distribution $\mu$, suppose we have a sampler following the distribution $\mu^{\ast j}\; (j\in \mathbb{N})$ and $\epsilon \ll\frac{1}{j}$, $\epsilon'>0$. Also, we assume that there exists a subset $A$ such that $\mathbf{0}\in A$ and $\sum_{\bfa\in A}\mu_{\bfa}\ge 1-\epsilon$.\\
    (i) Then $\mu^{\ast j}$ satisfies $\sum_{\bfa\in A'}\mu^{\ast j
    }\ge 1-j\epsilon$, with a subset $A'$ of size $\mathcal{O}(|A|^{j})$.\\
    (ii) Let $(\mu^{\ast j})'$ be the empirical distribution after the $N$-number of sampling from the $\mu^{\ast j}$. \\Then if $N\ge \mathcal{O}(|A|^{j}\epsilon^{'-2}\log(\delta_f^{-1}))$, then $\|\mu^{\ast j}-(\mu^{\ast j})'\|_1\le 3\epsilon'+2j\epsilon+\mathcal{O}(\epsilon'^2)$.\\
    (iii) If $N\ge \mathcal{O}(\epsilon^{-2}\log(\delta_f^{-1}))$, then $\|\mu^{\ast  j}-(\mu^{\ast  j})'\|_2\le \mathcal{O}(\epsilon)$.
    
\end{proposition}

The proof can be seen in App.~\ref{method:sampling_complexity}. 
Let $\eta(w,s)\equiv \max_{j}\left\{(c_j)^2\right\}$, and suppose $\sum_{\bfa\in A}p_{\bfa}\ge1-\frac{\epsilon}{4a(w,s)\sum_{j=1}^{a(w,s)}|c_j^{(w,s)}|}$ for some subset $A\subset \bfZ^n_2$ ($\Gamma(w,s)\equiv4a(w,s)\sum_{j=1}^{a(w,s)}|c_j^{(w,s)}|$ where $a(w,s)\le  \lfloor\frac{w+1}{2}\rfloor+(\lfloor\frac{w-1}{2}\rfloor+1)(w+s-1)$). We follow the above approximation, Eq.~\eqref{eq:method_approx_sol}. Then by Prop.~\ref{main:prop_sparsity_general}, with a failure probability $\delta_f$, $\mathcal{O}(a^3(w,s)\log(a(w,s)\delta_f^{-1})|A|^{2a(w,s)}\eta(w,s)\epsilon^{-2})$ number of samplings are enough to achieve
\begin{align}
    &\|p-p^{(w,s)}_{{\rm approx}}\|_1\le\nonumber\\&\sum_{j=1}^{a(w,s)}|c^{(w,s)}_j|\|\mu^{\ast j}-\mu'^{\ast j}\|_1+\mathcal{O}\left(\left(\frac{3w\delta}{2}\right)^{w+s}+\delta^w\right)\nonumber\\&\le\epsilon+\mathcal{O}\left(\left(\frac{3w\delta}{2}\right)^{w+s}+\delta^w\right).
\end{align}
The $a(w,s)^3\eta(w,s)$-term of sampling complexity is needed because we need an estimation accuracy $\epsilon'=\frac{\epsilon}{2}\left(a(w,s)\max_{j}(|c_j|)\right)^{-1}$ for each summation term, and $\mathcal{O}(a(w,s))$-number of input state is needed for single sample from $\mu^{\ast1},\mu^{\ast 2},\ldots, \mu^{\ast a(w,s)}$. Moreover, Prop.~\ref{main:prop_sparsity_general} (iii) implies that $\mathcal{O}\left(a^3(w,s)\eta(w,s)\epsilon^{-2}\log(a(w,s)\delta_f^{-1})\right)$, regarding $|A|=1$, is enough to obtain  $p_{\rm approx}$ with a similar bias with $l_2$ norm.

From the above argument and Prop.~\ref{main:prop_sparsity_general}, the next theorem summarizes the required sampling and time complexities to obtain $p'_{\rm approx}$ within a desired error $\epsilon$. 

\begin{theorem}[Sampling and time complexity of noise estimation]\label{thm:main_efficiency}
    Suppose that $\rho$ is an $n$-qubit noisy $k(\in \mathbb{N})$th-order hypergraph state. 
    Fix 
    the approximation order $w,s\leq 100$ and suppose the infidelity $\delta:= 1-\braket{\psi|\rho|\psi}<\frac{1}{3w}$. 
    Given $\epsilon>0$ and failure probability $\delta_f>0$, the following holds for the solution of Eq~\eqref{eq:thm_main_2markovian_eq}:\\ 
    (a) We can obtain the true noise distribution $p$ with a bias (with respect to $l_2$-norm) at most $\epsilon+\mathcal{O}\left(\left(\frac{3}{2}w\delta\right)^{w+s}+\delta^w\right)$  with $\mathcal{O}\left(\poly(\epsilon^{-1},\log\delta^{-1}_f)\right)$ sampling and $\mathcal{O}\left((2n)^{k}\poly(\epsilon^{-1},\log\delta^{-1}_f)\right)$ time complexity.\\
    (b) There exists a jointly increasing function $\Gamma(w,s)$ such that given  $\sum_{\bfa\in A}p_{\bfa}\ge 1-\frac{\epsilon}{\Gamma(w,s)}$ for some subset $A\subset \bfZ^n_2$, we can obtain $p$ with a bias $\epsilon+\mathcal{O}\left(\left(\frac{3}{2}w\delta\right)^{w+s}+\delta^w\right)$ in $l_1$-norm with $\mathcal{O}\left(\poly(\epsilon^{-1},\log\delta^{-1}_f,|A|^{(w^2+ws+w)})\right)$ sampling and $\mathcal{O}\left((2n)^{k}\poly(\epsilon^{-1},\log\delta^{-1}_f,|A|^{(w^2+ws+w)})\right)$ time complexity.\\
    In both (a) and (b), the complexity has a scaling factor depending on $(w,s)\in \mathbb{N}^2$.
  
\end{theorem}
Here, the bias is regarded as $\|p-p_{\rm approx}\|_{1\;(2)}$ between real solution $p$ and approximation $p_{{\rm approx}}$, respectively.
The bias of CPM depends on the degree of approximation order $(w,s)$. We need these two indices because we truncate the two parts of the infinite series component of the perturbed solution (see Eq.~\eqref{eq:after_inversion}). The additional $(2n)^{k}$ dependence of the time complexity is derived from the classically random selection of the tailoring gates and their merging into the post-processing gates (see Appendix~\ref{app:gate_merging} for details).

Theorem~\ref{thm:main_efficiency} states that dephasing noise $p$ of (tailored) hypergraph states can be efficiently approximated with much higher accuracy, in $l_2$-norm, compared to the case where we learn the Pauli noise itself within the infinite norm~\cite{harper2021}. The scaling factor in the complexity for CPM depends on $(w,s)$. Figure~\ref{fig:c_j_degree_approx} shows that as we increase the degree of approximation, the square of the maximal value of $c_j$ ($=\eta$) increases super exponentially with $j$. Nevertheless, we observe that within $w=8,s=5$, which shows sufficient accuracy in cases where $\delta\ll0.1$, the coefficient is within thousands.  

When $p$ has a sparse support~\cite {van2023,harper2021} ($|A|=\mathcal{O}(\poly(n))$), even the $l_1$ approximation can be done efficiently. Such a sparsity assumption is typically satisfied by many graph states of practical interest, where controlled-$Z$ gates are applied only between constant-sized qubits on a lattice. The reason is that noise at each gate propagates only within the connected qubits; after tailoring, the noise remains factorized. This behavior will also be demonstrated in the following numerical simulation based on a similar gate architecture. 

Moreover, this assumption still covers many practically relevant settings for universal fault-tolerant quantum computation. For example, graph states on two-dimensional square lattices are widely used in surface-code-based fault-tolerant quantum computation~\cite{fowler2012,google2023}, while three-dimensional graph states with constant-degree connectivity~\cite{herr2018,bergamaschi2025} serve as resource states for fault-tolerant MBQC and quantum low-density-parity-check codes. Lastly, a logically encoded Union Jack state~\cite{miller2016}, which is a third-order hypergraph state with constant connectivity, can also be used for universal MBQC by using Pauli measurement only.

Even though the $l_1$ accuracy is not guaranteed when the sparsity is broken, the CPM algorithm itself can be executed within bounded sampling. Moreover, it could be possible to normalize the obtained solution to a non-sparse solution under additional assumptions about the noise structure (i.e., white noise~\cite{tsubouchi2024}, thermal noise~\cite{oliver1965}, etc.).

\begin{figure}[t]
	\includegraphics[width=8cm]{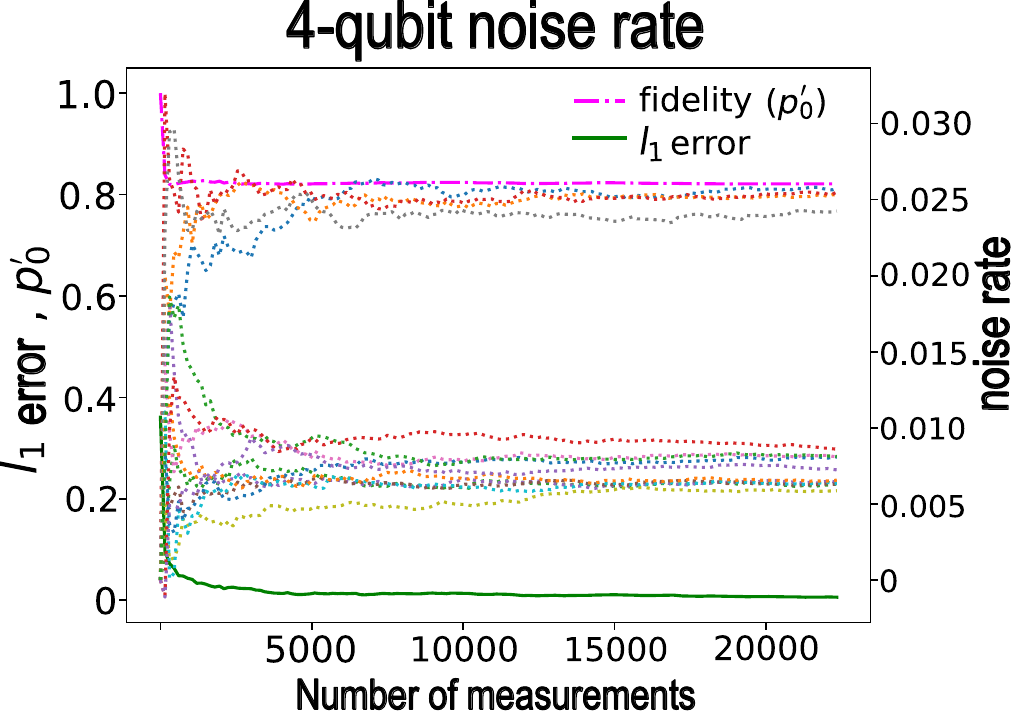}

	\caption{Estimation accuracy of $4$-qubit noise learning of third-order complete hypergraph state $\ket{K_4}$. We arbitrarily chose the noisy input state with $p_0\simeq0.819$. For each sampling step, we reconstruct $p'$ via Eq.~\eqref{eq:main_exactsolution_fwht}. The green solid line is the $l_1$-distance between the given step's $p'$ and $p$. The dashed and dash-dotted lines correspond to the $16$ components of $p'$.
	}\label{fig:4_qubit_numerical_random_noise}
\end{figure}

\begin{figure}[t]
    \centering
    
    \includegraphics[width=\columnwidth]{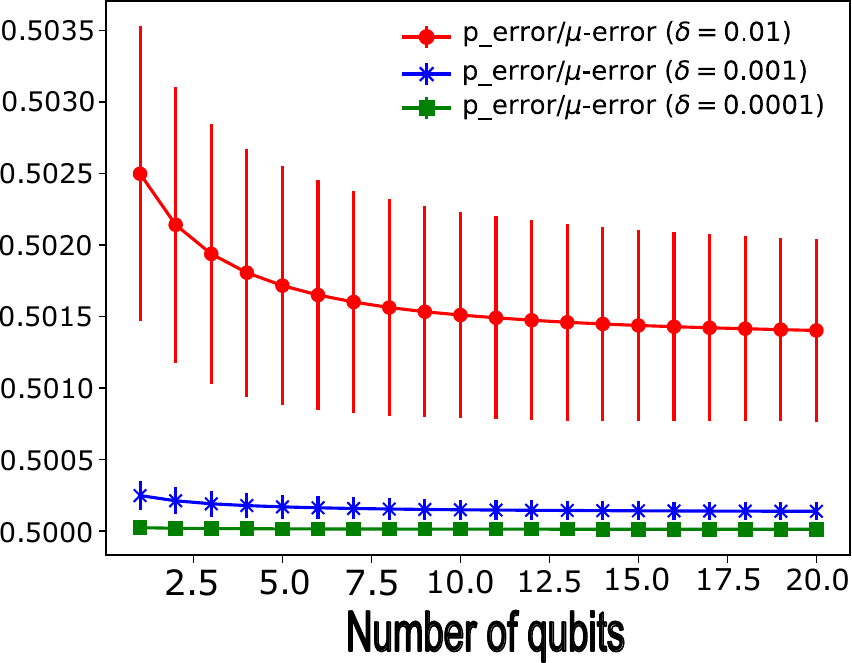}

    \caption{The averaged value of $\frac{\|p^{(1)}-p^{(2)}\|_1}{\|\mu^{(1)}-\mu^{(2)}\|_1}$. We randomly chose $\mu^{(1)}_{\mathbf{0}},\mu^{(2)}_{\mathbf{0}}\ge 1-\delta$, and other noise rates are uniformly chosen on the probability simplex (Dirichlet distribution~\cite{mackay2003}). Then we calculate quasi-probabilities $p^{(1)},p^{(2)}$ via $p=\frac{1}{2^n}H\sqrt{H\mu}$. We averaged the rate over $5000$ samples. The x-axis represents the number of qubits $n$. 
    }\label{fig:numerical_dirichlet_noise}
    \end{figure}

\section{Numerical results}

In this section, we demonstrate numerical simulations showing the accuracy and efficiency of our noise detection algorithm. Figure~\ref{fig:4_qubit_numerical_random_noise} plots the estimated $p'$ of the noisy $4$-qubit complete hypergraph state, using the FWHT method. Given $p$ (noise distribution) and $\mu$ (measurement distribution), our method successfully estimates $p'$ within a small $l_1$-error. Furthermore, in Fig.~\ref{fig:numerical_dirichlet_noise}, we have estimated the rate between the distance between two different measurement distributions $\|\mu^{(1)}-\mu^{(2)}\|_1$ and the distance between two following solutions $\|p^{(1)}- p^{(2)}\|_1$. We can see that when the noise rate $\delta$ is small, the average rate converges to $\simeq 0.5+\delta$ by increasing $n$. Hence we expect that Eq.~\eqref{eq:asymptotic_p_dist} holds even if we set $\epsilon_{\rm th}=\mathcal{O}(1)$. In other words, it implies that even for large $n$, the scaling of $l_1$-distance between $p$ and $p'$ follows that of $\mu$ and $\mu'$. 

\begin{figure}[t]
	\centering
	
	\includegraphics[width=\columnwidth]{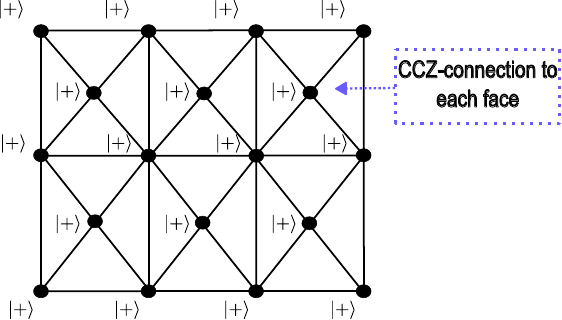}

	\caption{Schematic diagram of $18$-qubit Union Jack state. Each vertex represents a $\ket{+}$-state and the \textsc{CCZ}-gate acts on each face. }\label{fig:union_jack_diagram}
\end{figure}

\begin{figure}[t]
	
	\centering
	
	\includegraphics[width=\columnwidth]{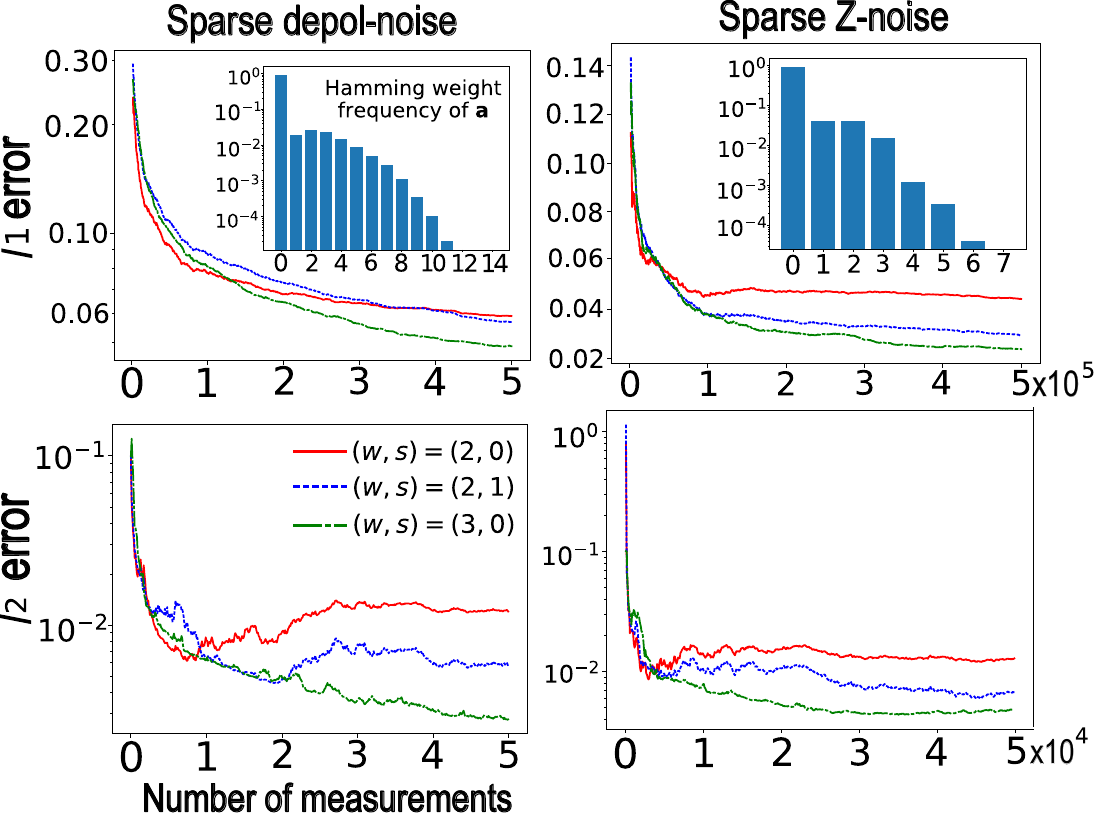}

	\caption{Estimation accuracy of noise learning to noisy (tailored) $18$-qubit Union Jack state. Here, the x-axis indicates the number of measurements to obtain the empirical distribution $(\mu^{\ast0})',(\mu^{\ast1})',(\mu^{\ast2})',(\mu^{\ast3})'$ of true one ($\mu$). For each number of measurements, we used the $(w,s)=(2,0)$ ($(w,s)=(2,1)$, $(w,s)=(3,0)$ resp.)-biased  approximations (Eq.~\eqref{eq:approx_solution_proto} and Eq.~\eqref{eq:approx_solution_proto2} in the main text) to get the estimated solution $p'$. The number of samples is counted as $\times 4$ ($6,8$) of the number of measurements. We gave uniformly random Pauli noise from the 3-qubit Pauli group (Pauli-$Z$ subgroup resp.) with a probability $0.005$, resulting  $\delta\simeq 0.099$ ($0.096$), on each $CCZ$ gate. Upper graphs measure $l_1$ distance, and lower graphs measure $l_2$ distance between $p$ (true solution) and $p'$. Small sub-graphs indicate the frequency of Hamming weight of measurement outcome following the distribution $p_{\bfa}=\braket{\psi_{\bfa}|\rho|\psi_{\bfa}}$. We note that the y-axis of the Hamming weight frequency has a log scale.}
	\label{fig:union_jack_approx}
\end{figure}

Next, let us look at Fig.~\ref{fig:union_jack_approx}. Here we employed more practical quantum noises, depolarizing noise, and $Z$-noise. We chose the target as $18$-qubit Union Jack states~\cite{zhu2019} (Fig.~\ref{fig:union_jack_diagram}), and gave the depolarizing [or Pauli $Z$] noise to CCZ gates. Here, we found the true distribution $p$ via 50000 measurement samples of $\bfa$ following $p_{\bfa}=\braket{\psi_{\bfa}|\rho|\psi_{\bfa}}$. We adopted such an approximate calculation because our algorithm requires a $(18\times 2)=36$-qubit system, which takes a lot of time when using a desktop computer. Nonetheless, this approximation is sufficiently accurate since the noise is sparse, i.e., $p_{\bfa}$ of large Hamming weight is exponentially small. 
Again, the $2$nd, $3$rd-degree approximation methods successfully estimate the noise distribution $p$ within a small $l_2$ (or $l_1$) error much lower than $\delta$. As implied by the statement of Thm.~\ref{thm:main_efficiency} (a), $l_2$ error in both cases decreases below $10^{-2}$ within $10000$ copies, regardless of the noise sparsity.  
$3$rd-degree approximation performs better than the $2$nd one. 
Notably, for the $l_1$-error, we observe that the $Z$-noise case exhibits the best performance in terms of sampling complexity. The reason is that compared to the depolarizing noise, $Z$-noise commutes with every \textsc{CCZ} gate and does not propagate to other qubits. It means that the frequency of largely weighted measurement flip-error is much lower, compared to the other cases, allowing a small size of $|A|$ in Thm.~\ref{thm:main_efficiency} (b). 

\section{Discussion}

We have introduced a noise tailoring method and how to efficiently learn the tailored dephasing-noise rate (diagonals) for graph and hypergraph states. We designed the quantum algorithm to either exactly (WHT) or approximately (CPM) obtain a tailored dephasing noise rate. Assuming the infidelity $\delta<\frac{1}{3w}$, the $l_2$-approximation is efficient in both time complexity (polynomial in $n$) with a fixed approximation order $(w,s)$, and circuit depth (e.g., single-depth \textsc{CNOT} is sufficient for graph states). Building on existing works on Pauli noise learning, which achieve efficient sample and time complexity with respect to the infinity norm $\|\cdot\|_{\infty}$~\cite{harper2021,flammia2020}, our result extends this efficiency guarantee to the $l_2$ norm, which can be exponentially larger than the infinity norm in the worst case. Though in our work, we focused on the estimation of tailored noise in graph states rather than the noisy channel, bypassing the obstacles of efficiency (i.e., WHT after Pauli eigenvalue estimation~\cite{chen2023,chen2024}).

Estimating tailored noise offers richer information compared to the previous state verification or randomized benchmarking~\cite{zhu2019,liu2024,lee2025_benchmarking}, and enables us to apply this gained information to post-error mitigation~\cite{lostaglio2021,suzuki2022}, or to introduce a better magic state preparation platform~\cite{gupta2024}. Consequently, these would lead to improved accuracy in graph-based algorithms such as quantum metrology~\cite{shettell2020}, MBQC~\cite{miller2016,briegel2009}, surface code~\cite{aharonov2025,suzuki2022}, and instantaneous quantum polynomial (IQP) circuit~\cite{rajakumar2024}. Our work implies that quantum computing based on a largely entangled resource state~\cite{takeuchi2019} can circumvent the conventional problem of the inefficiency in noise learning and recovery~\cite{flammia2020,berg2023,harper2021}. We note that sparsity (poly-support of $p$) implies the poly-rank of the noisy state $\rho$ ($\because$ Prop.~\ref{prop:dominant_rank}). Therefore,  the $l_1$-approximation efficiency of the algorithm for sparse $p$ also contributes to the primary goal of compressive sensing~\cite{gross2010}, another research field to identify the unknown states with poly-rank. Although poly-rank does not imply a sparsity of $p$, we believe that our scheme has the potential to be generalized to $l_1$-approximation of the noise in low-rank hypergraph states.

Given that the WH spectrum of $p$ is positive, the CPM-solution becomes accurate as we increase the order $(w,s)$ and sampling-copy number $N$. However, if the input state is very noisy ($\delta>\frac{1}{3w}$), we cannot guarantee that this approximation method is still efficient. Hence, further refinement to encompass a wider efficient region should be an important area of work, while we leave it as an open problem. A more challenging task is to generalize our scheme to estimate the off-diagonal elements of the input state $\rho$. Non-linearity of the noise equation gives another open problem. To be specific, we could interpret it as the cost of trying to detect the noise of magic states using the resources of a one-step lower Clifford hierarchy~\cite{anderson2024}. In future work, we expect to find an efficient stabilizer circuit scheme for general noise and to completely formalize the trade-off among the nonlinearity of the noise equation, computational complexity, and the hypergraph-state hierarchy.

\vspace{5pt}

\noindent
        \textbf{*Note:}
    We became aware of a recent work~\cite{zenrui2025} that introduces an alternative approximation technique for the diagonal elements of the input state. This work also applies Bell sampling and WHT, whereas the WHT instead targets expectation values of the stabilizer operators associated with the given pure graph state.
        
\vspace{5pt}

\section{Acknowledgments}

The authors thank Myungshik Kim for the thoughtful discussions throughout the project. This work was supported by the Korean government [Ministry of Science and ICT (MSIT)], the NRF grants funded by the Korea government (MSIT) (No. RS-2024-00413957, No. RS-2024-00438415, and RS-2023-NR076733), the Institute of Information \& Communications Technology Planning \& Evaluation (IITP) grant funded by the Korea government (MSIT) (IITP-2025-RS-2020-II201606 and IITP-2025-RS-2024-00437191), and the Institute of Applied Physics at Seoul National University. This research was also supported by the education and training program of the Quantum Information Research Support Center, funded through the National research foundation of Korea~(NRF) by the Ministry of science and ICT(MSIT) of the Korean government~(No.2021M3H3A1036573). J.S. thanks support from the Innovate UK (Project No.10075020) and Schmidt Sciences, LLC.

\appendix

\section{Clifford hierarchy of multiple controlled-Z gates}~\label{app:CCZ_hierarchy}
This section shows that the $k$th-order multiple controlled-Z gate is of $k$th-order Clifford hierarchy. We first need the following lemma. 

\begin{lemma}\label{lem: P_tensor_hierarchy}
    Suppose $U$ be an $n$-qubit unitary of $k$th-order Clifford hierarchy and $P$ be a single-qubit Pauli operator. Then the following holds:
    (i) The $n+1$-qubit unitary $P\otimes U$ is of $k$th-order Clifford hierarchy.\\
    (ii) If $k>1$, $UP,PU\in {\rm Cl}^{(k)}_n$.
\end{lemma}
\begin{proof}
    (i) We leave the proof for the case $k\le 2$. We assume it holds for $k=k'>2$. Then for arbitrary $P'\in \mathcal{P}_{n+1}$ and $U\in {\rm Cl}_n^{(k'+1)}$, $(P\otimes U)P'(P\otimes U)^{\dagger}=Q\otimes UQ'U^{\dagger}$ for some $Q\in \mathcal{P}_1$ and $Q'\in \mathcal{P}_{n}$.  Note that $U Q'U^{\dag}\in {\rm Cl}_{n}^{(k')}$ and hence by the assumption $Q\otimes UQ'U^{\dag}\in {\rm Cl}_{n}^{(k')}$. The proof is completed by induction.\\
    (ii) Similarly, suppose the statement holds for $(k-1)$th-order ($k>2$, proving for $k=2$ case is trivial). Then for an arbitrary Pauli operator $Q\in \mathcal{P}_n$, $PUQU^{\dagger}P^{\dagger}=PU'P$ for some unitary $U'\in {\rm Cl}^{k-1}_n$. Then by the hypothesis, $PU'\in {\rm Cl}^{(k-1)}_n$ and $(PU')P\in {\rm Cl}^{(k-1)}_n$. Therefore, $PU\in {\rm Cl}^{(k)}_n$. A similar manner applies to show for $UP$.  
\end{proof}

Now, we solve the main problem:
\begin{proposition}\label{thm:hypergraph_hierarchy}
    Given a $k$th-degree hypergraph $G(V,E)$, $\prod_{A\in E}C_{A}Z$ is a unitary of $k$-th ordered Clifford hierarchy.
\end{proposition}

\begin{proof}
    We already noted that if $k\le 2$, $\Pi_{A\in E}C_{A}Z$ must be in the $k$-th order Clifford hierarchy. Now, we assume it also holds when the maximal order of $E$ is $k'>2$. We will prove that it leads to a similar statement when the order is $k'+1$. We first show that $C_{A}Z$ with $|A|=k'+1$ is of $(k'+1)$th-order Clifford hierarchy. Given an element $a'\in A$, $C_{A\backslash\left\{a'\right\}}Z$ is of $k'$-th Clifford hierarchy. Without losing generality, we assume $a'=1$. Note that $C_{A}Z$ commutes with the Pauli $Z$ operator. Furthermore, for $\bfx\in \bfZ^n_2$ (note that $C_{A}Z$ is Hermitian),
    \begin{align}
    &C_{A}Z (X\otimes I\otimes \ldots\otimes I)C_{A}Z\ket{\bfx}\nonumber\\&=C_{A}Z (-1)^{\prod_{a\in A}x_{a}}\ket{x_{1}+1,x_2,\ldots,x_n}
       \nonumber\\&
        =(-1)^{\Pi_{a\in A}x_{a}+\Pi_{a\in A}(x_{a}+\delta_{a,1})}\ket{x_{1}+1,x_2,\ldots,x_n}  
        \nonumber\\&=(-1)^{\delta_{A}(1)\prod_{a\in A\backslash\left\{1\right\}}x_{a}}\ket{x_{1}+1,x_2,\ldots,x_n}
        \nonumber\\&
        =\left\{X\otimes \left(\delta_{A}(1)C_{A\backslash\left\{a'\right\}}Z+(1-\delta_{A}(1))I^{\otimes n-1}\right)\right\}\ket{\bfx}\label{eq:hyper_cz_twirling},
    \end{align}
    where $\delta_A(i)=1\; (i\in \mathbb{N})$ if $i\in A$ and zero otherwise.
    Therefore, we conclude that $C_{A}Z (X\otimes I\otimes \ldots\otimes I)C_{A}Z=(X\otimes C_{A\backslash\left\{a'\right\}}Z)\in {\rm Cl}_{n}^{(k-1)}$ ($\because$ Lem.~\ref{lem: P_tensor_hierarchy} (i), $\mathcal{P}_n\subset{\rm Cl}_n^{(k-1)}$~\cite{anderson2024}). From these two facts and that an arbitrary Pauli operator can be decomposed as a product of single-qubit $\pm iZ,\pm iX$ operators, we conclude that $C_{A}Z P C_{A}Z$ is a product of multiple controlled  Z gates of order $k'$ and some Pauli operator, which is contained in 
    ${\rm Cl}_{n}^{k-1}$ ($\because$ Lem.~\ref{lem: P_tensor_hierarchy} (ii))· 
    
    The $C_{A\backslash \left\{a\right\}}Z$ commutes with all diagonal gates. Even if we twirl another $C_{A'}Z\;(A'\in E)$ gates, the non-trivial effect occurs only on the Pauli operators. Then we can use the same logic and conclude that $\Pi_{A\in E}C_{A}Z$ is also $k$th-order hierarchy given that the maximal size of $A$ is $k$.  

    Moreover, each Pauli operator generated by each $C_{A}Z$ can be propagated forward to other $C_{A}Z$ gates, leaving another Pauli operator and multiple controlled gates with the order lower by one. However, these controlled-Z gates also commute with $C_{A}Z$ gates. Therefore, a depth larger than one is formed by only multiple controlled-Z gates which are generated during the propagation. The following result is that the maximal circuit depth for the noise tailoring of third-order hypergraph states is equal to the maximal depth of \textsc{CZ}-gates, $n+1$.

\end{proof} 

\section{Gate merging}\label{app:gate_merging}

We observe that our circuit has three sections providing multiple controlled-$Z$ operations, tailoring of two copies, and the feedback operation. However, we can merge three non-local parts into one. The non-local part of the first copy tailoring is easy because all controlled $Z$ gates can commute with the control qubits of transversal \textsc{CNOT} gates, merging into the feedback section. 

The second tailoring part is a bit trickier. We set $\ket{\bfx,\bfu'}\;(\bfx,\bfu'\in \bfZ^{n}_2)$ as computational basis of two-copies. We denote $\textsc{CNOT}_{||}$ as the operator of transversal \textsc{CNOT} gates and $\textsc{CZ}^{(\bfa,t)}$ as controlled gates for tailoring operation $X_\psi^{\bfa}$, where the phase polynomial is $T(\bfx)\;[\textsc{CZ}^{(\bfa,t)}\ket{\bfx}=(-1)^{T(\bfx)}\ket{\bfx}]$ to the second copy. Then we obtain that 
\begin{align}
    &\textsc{CNOT}_{||}\textsc{CZ}^{(\bfa,t)} \textsc{CNOT}_{||}\ket{\bfx,\bfu'}\nonumber\\&=\textsc{CNOT}_{||}\textsc{CZ}^{(\bfa,t)} \ket{\bfx,\bfx+\bfu'}\nonumber\\&=(-1)^{T(\bfx+\bfu')}\ket{\bfx,\bfu'}\nonumber\\&=(-1)^{T(\bfx)+T(\bfu')+\left\{T(\bfx+\bfu')+T(\bfx)+T(\bfu')\right\}}\ket{\bfx,\bfu'}.
\end{align}
Therefore, when the second tailoring gates commute with the transversal \textsc{CNOT} gates, it outputs the same gates to the first copy, and controlled gates whose phase polynomial is $T(\bfx)+T(\bfu')+T(\bfx+\bfu')$. The remaining $T(\bfu')$ is absorbed in the last computational basis measurement of the second copy, and shall not be counted. $T(\bfx)+T(\bfu')+T(\bfx+\bfu')$ does not have any monomials solely containing $\bfx$ or $\bfu'$ variables. Without losing generality, we only consider the multiple controlled-Z gate $C_AZ$ for some subset $A$ which gives $C_AZ(\ket{\bfx,\bfu'})=(-1)^{x_1x_2,\ldots x_vu_1u_2,\ldots u_w}\ket{\bfx,\bfu'}\;(v,w\in [n])$. The point is that after we enact such a gate, we measure the second copy. Specifically, the effective operation on an arbitrary state $\sigma$ is expressed by 
{\small
\begin{align}
    &\sigma\leftarrow \sum_{\bfu',\bfk_1,\bfk_2,\bfl_1,\bfl_2\in \bfZ^n_2}\sigma_{\bfk,\bfl}\textsc{CZ}^{(\bfu')}(I\otimes \ket{\bfu'}\bra{\bfu'})C_AZ\nonumber\\&\ket{\bfk_1\bfk_2}\bra{\bfl_1\bfl_2}C_AZ(I\otimes \ket{\bfu'}\bra{\bfu'})\textsc{CZ}^{(\bfu')}\nonumber\\&=\sum_{\bfu',\bfk_1,\bfk_2,\bfl_1,\bfl_2\in \bfZ^n_2}\sigma_{\bfk,\bfl}\textsc{CZ}^{(\bfu')}(I\otimes \ket{\bfu'}\bra{\bfu'})\nonumber\\&(-1)^{\bfk_{11}\bfk_{12},\ldots \bfk_{1v}\bfk_{21}\bfk_{22},\ldots \bfk_{2w}}\ket{\bfk_1\bfk_2}\bra{\bfl_1\bfl_2}\nonumber\\&\times(-1)^{\bfl_{11}\bfl_{12},\ldots \bfl_{1v}\bfl_{21}\bfl_{22},\ldots \bfl_{2w}}(I\otimes \ket{\bfu'}\bra{\bfu'})\textsc{CZ}^{(\bfu')}\nonumber\\&=\sum_{\bfu',\bfk_1,\bfl_1\in \bfZ^n_2}\sigma_{(\bfk_1,\bfu'),(\bfl_1,\bfu')}\textsc{CZ}^{(\bfu')}(-1)^{\bfk_{11}\bfk_{12},\ldots \bfk_{1v}\bfu'_{1}\bfu'_{2},\ldots \bfu'_{yw}}\nonumber\\&\ket{\bfk_1\bfu'}\bra{\bfl_1\bfu'}(-1)^{\bfl_{11}\bfl_{12},\ldots \bfl_{1v}\bfu'_{1}\bfu'_{2},\ldots \bfu'_{yw}}\textsc{CZ}^{(\bfu')},
\end{align}
}where $\textsc{CZ}^{(\bfu')}$ is the post-processing multiple controlled-$Z$ operation in the original feedback section, and $\bfk (\bfl\;{\rm resp.})=(\bfk(\bfl)_1,\bfk(\bfl)_2)$. The final expression is equivalent to the additional multiple controlled-$Z$ operation to the first copy, post-processed by the joint information of outcome $\bfu'$. Furthermore, such additional multiple controlled-$Z$ operations reside on a lower Clifford hierarchy, compared to $C_A Z$. In conclusion, all multiple-controlled $Z$ gates over two copies can be merged into the feedback section. 

We now discuss the time complexity of computing such gate merging for $k\in \mathbb{N}$-th order state. Both choosing tailoring gates $(X_\psi^{\bfa_1},X_\psi^{\bfa_2})$, $CZ^{(\bfu')}$ after the second-copy measurement, and gate merging involve the calculation of directional derivative $T(\bfx|\bfu'):=T(\bfx+\bfu)+T(\bfx)$. Assuming the degree of $T$ is $k$, we find $T(\bfx|\bfu')$ as follows: for each monomial $x_{i_1}x_{i_2}\ldots x_{i_k}$ in $T$, we calculate $(x_{i_1}+u'_{i_1})(x_{i_2}+u'_{i_2})\ldots(x_{i_k}+u'_{i_k})+x_{i_1}x_{i_2}\ldots x_{i_k}$ within $\mathcal{O}(2^k)$ time. Repeating this procedure for all monomials in $T$, we can find $T(\bfx|\bfu')$ within $\mathcal{O}(2^kn^{k})$ time. We also need to record the coefficients of generating monomials in $\mathcal{O}(n^{k})$ memory so that we can flip the coefficient when the same monomial is generated in the next steps.

When $k=3$, we only need \textsc{CZ} operation on the feedback section, and then only $\frac{n}{2}+\mathcal{O}(\log(n))$-depth is required~\cite{maslov2022}. Gate count is at most $\comb{n}{2}=\mathcal{O}(n^2)$, but the average value is lower:

\begin{proposition}
    The number of two-qubit Clifford gates to run the noise tailoring and learning of third-order hypergraph states is $\mathcal{O}(\frac{n^2}{2})$, where the average over measurement copies is at most $\Theta (\frac{n^2}{4})$.
    
\end{proposition}
\begin{proof}
      The maximal number $\mathcal{O}(\comb{n}{2})=\mathcal{O}(\frac{n^{2}}{2})$ is definite~\cite{maslov2022}. So, we only prove the average number. Let $T$ be the phase polynomial for the target third-order hypergraph state $\ket{\psi}$. We claim that $T(\bfx)+T(\bfx+\bfa_1)+T(\bfx+\bfa_2)+T(\bfx+\bfa_1+\bfa_2)$ is an affine function with the variable $\bfx$. We prove it by checking a monomial $x_ix_jx_k$ satisfies 
    \begin{align}
    	&x_ix_jx_k+(x_i+a_{1i})(x_j+a_{1j})(x_k+a_{1k})\nonumber\\&+(x_i+a_{2i})(x_j+a_{2j})(x_k+a_{2k})+(x_i+a_{1i}\nonumber\\&+a_{2i})(x_j+a_{1j}+a_{2j})(x_k+a_{1k}+a_{2k})\nonumber\\&=f_i(\bfa_1,\bfa_2)x_i+f_j(\bfa_1,\bfa_2)x_j+f_k(\bfa_1,\bfa_2)x_k+g(\bfa_1,\bfa_2),
    \end{align} where $f_i,f_j,f_k$ ($g$ resp.) are  quadratic (cubic) function of $\bfa_1,\bfa_2$. Therefore, since $T$ is a sum of monomials, we conclude that the quadratic term of $T(\bfx+\bfa_1)+T(\bfx)+T(\bfx+\bfa_2)+T(\bfx)$, which determines the tailoring-$CZ$ gates in the merged section, is equal to the quadratic part of $T(\bfx+\bfa_1+\bfa_2)+T(\bfx)$. 
    
    Second, after the merging, we take the measure with the computational basis in the second copy, to $CNOT_{||}(X^{\bfa_1}\otimes X^{\bfa_2})\rho^{\otimes 2} (X^{\bfa_1}\otimes X^{\bfa_2})CNOT_{||}$, and obtain the outcome $\bfu'$. The outcome $\bfu'$ is sampled by some probability $q(\bfa_1+\bfa_2+\bfu')$. To be specific, $q$ depends solely on $\rho\otimes \rho$ and is translated by $X^{\bfa_1}\otimes X^{\bfa_2}$. The post-processing $CZ$ operation by the outcome $\bfu'$ is also represented as quadratic terms in  $T(\bfx+\bfu')+T(\bfx)$. By the previous argument, $CZ$ gates of the merged section are equal to the quadratic terms in $T(\bfx+\bfa_1+\bfa_2)+T(\bfx+\bfu')$. By taking the translation, $T(\bfx)+T(\bfx+\bfa_1+\bfa_2+\bfu')$ has the same number of quadratic terms. We recall that the probability to sample $\bfu'$ conditioned on uniformly sampled $(\bfa_1,\bfa_2)$ is $q(\bfa_1+\bfa_2+\bfu')$. Translating $\bfu'\leftarrow \bfa_1+\bfa_2+\bfu'$, we conclude that the number of quadratic terms is determined solely by the second-copy measurement distribution $q$.

    Next, we denote $\bar{p}:\bfZ_2^n\rightarrow \mathbb{R}$ such that $\bar{p}_{\bfx}\equiv \braket{\bfx|\rho|\bfx}$, diagonal elements with the computational basis. Then we obtain that 
    \begin{align}
        q_\bfx&=\sum_{\bfy\in \bfZ_2^n}\braket{\bfy\bfx|CNOT_{||}\rho^{\otimes 2}CNOT_{||}|\bfy\bfx}\nonumber\\&=\sum_{\bfy\in \bfZ_2^n}\braket{\bfy,(\bfx+\bfy)|\rho^{\otimes 2}|\bfy,(\bfx+\bfy)}=\bar{p}^{\ast 2}_\bfx.
    \end{align}
    Therefore, $q$ is the self-convolution of the diagonal element of $\rho$. Furthermore, let us divide $\bfx$ into two parts $(\bfx_1,\bfx_2)$. Then the marginal probability is, 
    {\small
    \begin{align}\label{eq:marginal_self_convolution_bar}
        &q_{\bfx_1}:=\sum_{\bfx_2}q_{(\bfx_1,\bfx_2)}=\nonumber\\&\sum_{\bfx_2,\bfy_1,\bfy_2}\braket{(\bfy_1,\bfy_2),(\bfy_1+\bfx_1,\bfx_2)|\rho^{\otimes 2}|(\bfy_1,\bfy_2),(\bfy_1+\bfx_1,\bfx_2)}\nonumber\\&=(\bar{p}^{(1)\ast 2})_{\bfx_1},
    \end{align}
    }where the last form $\bar{p}_{\bfb}^{(1)}:=\sum_{\bfx_2}\bar{p}_{(\bfb,\bfx_2)}$. It means that the marginalization still follows the parallel rule.

    Now we scrutinize the quadratic terms in $T(\bfx+\bfu)+T(\bfx)$. Given $\ket{\psi}=\ket{G(V,E)}$, we define the 2-section graph $G'(V,E')$ such that 
    \begin{align}
    E'\equiv \left\{\left\{u,v\right\}|\exists e\in E, \left\{u,v\right\}\subset e\right\}.
    \end{align}
    We note that the 2-section graph $G'$ is a (second-order) graph, where the edges of $G'$ are the opposite edges of each vertex in the original hypergraph edge.

    Following the above notation, when $u_i=1$ is sampled, we enact the $\textsc{CZ}$-gates following the opposite edges of hyperedges that contain $i$-th vertex. Equivalently, each CZ-operation on the edge of $E'$ is determined by the parity of the number of selected opposite vertices after the sampling of $\bfu$ from $q$. 
    
    Moreover, we can elaborate the expression of the probability of the parity. To do so, we denote $q^{(A)}$ as marginal probability over set $A\subset [n]$ of $q$  (with a similar definition to $\bar{p}^{(A)}$). We also denote $P(\mathbf{m})\;(\mathbf{m}\in \bfZ^{|E'|}_2)$ as the probability to get the parity vector $\mathbf{m}\in \bfZ_2^{|E'|}$ over $E'$. Its marginal probability to get odd parity on the single $i(\in[|E'|])$th edge whose set of opposite vertices is $A_i$, say $P(m_i=1)$, is obtained by Eq.~\eqref{eq:marginal_self_convolution_bar}:
    \begin{align}
        P(m_i=1):=\sum_{\bar{\bfu}\in A_i,|\bar{\bfu}|{\;\rm is\; odd}}q^{(A)}_{\bar{\bfu}}=\frac{1-\widehat{\bar{p}^{(A_i)}}_{(1,1,\ldots,1)}^2}{2},
    \end{align}
    where the last form is derived by the fact that
    \begin{align}
        \widehat{\bar{p}^{(A_i)}}_{(1,1,\ldots,1)}^2=\widehat{q^{(A_i)}}_{(1,1,\ldots,1)}=1-2\sum_{\bar{\bfu}\in A,|\bar{\bfu}|{\;\rm is\; odd}}q^{(A_i)}_{\bar{\bfu}}.
    \end{align}

    Finally, the average number of \textsc{CZ}-gates for the merged section in the circuit is 
    
    \begin{align}
        \sum_{\mathbf{m}\in \bfZ^{|E'|}_2}\left(\sum_{i=1}^{|E'|}m_i\right)P(\mathbf{m})&=
        \sum_{i=1}^{|E'|}\sum_{\mathbf{m}\in \bfZ^{|E'|}_2,m_i=1}P(\mathbf{m})\nonumber\\&=\sum_{i=1}^{|E'|}P(m_i=1)\nonumber\\&=\sum_{i=1}^{|E'|}\frac{1-\widehat{\bar{p}^{(A_i)}}_{(1,1,\ldots,1)}^2}{2}\nonumber\\&\le\frac{|E'|}{2}\;\left(\because \left|\widehat{\bar{p}^{(A_i)}}_{(1,1,\ldots,1)}\right|\le 1\right).
    \end{align}
     Conclusively, given that $|E'|\le\comb{n}{2}\simeq \frac{n^2}{2}$,  the average number of two-qubit gates for the noise convolution is at most $\frac{n^2}{4}$. We can see that it hits the maximum when $\ket{\psi}$ is the third-order complete hypergraph state, and $\forall \bfx\in \bfZ_2^n,\; \braket{\bfx|\rho|\bfx}=\frac{1}{2^n}\;(\Rightarrow \forall \widehat{\bar{p}^{(A_i)}}_{(1,1,\ldots,1)}=0)$ that holds when $\rho$ is subjected to a dephasing noise. 
\end{proof}

\section{Validity and complexity of WHT}\label{app_time_solving_noise_equation}

We recall that the solution of the convoluted noise equation is expressed exactly as $p=\frac{1}{2^n}H\sqrt{H\mu}$. We note that 
\begin{align}\label{eq:eigenvector_P}
    \sum_{\bfb}p_{\bfa+\bfb}(-1)^{\bfb\cdot\bfc}&=\sum_{\bfb}p_{\bfb}(-1)^{\bfb\cdot \bfc}\times (-1)^{\bfa\cdot \bfc}\nonumber\\&\Rightarrow P(p)\eta^{(\bfc)}=\widehat{p}_{\bfc}\eta^{(\bfc)},
\end{align}
where $\eta^{(\bfc)}_{\bfb}\equiv (-1)^{\bfb\cdot \bfc}$. Therefore, each $\eta^{(\bfc)}$ is the (unnormalized) eigenvector of $P(p)$ with the eigenvalue $\widehat{p}_{\bfc}$.

In the main text, we saw that $\widehat{\mu}$ is always non-negative, and $\widehat{p}$ is also non-negative given that $\delta<\frac{1}{2}$. Considering Eq.~\eqref{eq:eigenvector_P}, we expect that there are many cases where FWHT is well-defined even if $\delta\ge \frac{1}{2}$. For example, we have the following result. 

\begin{proposition}
    Suppose that a graph state $\ket{\psi}$ is generated by layer of $\textsc{CZ}$ gates, $L_i\;(i\in [d])$ for some $d\in \mathbb{N}$. In other words, $\ket{\psi}=\prod_{i=1}^{d}L_i\ket{+}^{\otimes n}$. We also assume that each $L_i$ has the Pauli noise $\mathcal{N}_i$ such that $\mathcal{N}_i(\cdot)=\sum_{\bfa\in \bfZ^{2n}_2}p^{(i)}_\bfa T_{\bfa}(\cdot)T_{\bfa}$.  If all $2^{2n}$ by $2^{2n}$ matrix $\widetilde{P}^{(i)}$, which is defined by $\widetilde{P}^{(i)}_{\bfa,\bfb}=p^{(i)}_{\bfa+\bfb}$, are non-negative, then $P(p)\ge 0$ ($\Rightarrow \widehat{p}\ge 0$).  
\end{proposition}

\begin{proof}
    The main proof idea is that we first back-propagate each noise channel to the right of $\ket{+}^{\otimes n}$. Then take the $Z$-noise marginals of the merged noise channel. 
     In this setup, the marginal $Z$-noise of total Pauli noise $\mathcal{N}(\cdot)=\sum_{\bfa\in \bfZ^{2n}_2}\widetilde{p}_{\bfa}T_{\bfa}(\cdot)T_{\bfa}$ is tailored noise. We first prove that if   $\widetilde{P}_{\bfa,\bfb}\equiv \widetilde{p}_{\bfa+\bfb}$ and $\widetilde{P}\ge 0$, then $P\ge 0$. This is because,
    \begin{align}
    P_{\bfa,\bfb}=p_{\bfa+\bfb}=\sum_{\bfc\in \bfZ^n_2}\widetilde{p}_{(\bfc,\bfa+\bfb)}&= \frac{1}{2^n}\sum_{\bfc,\bfc'\in \bfZ^n_2}\widetilde{p}_{(\bfc+\bfc',\bfa+\bfb)}\nonumber\\&=\frac{1}{2^n}\sum_{\bfc,\bfc'\in \bfZ^n_2}\braket{\bfc,\bfa|\widetilde{P}|\bfc',\bfb}.
    \end{align}
    
    Therefore, for an arbitrary vector $\ket{\tau}=\sum_{\bfa\in \bfZ^n_2}c_{\bfa}\ket{\bfa}$, we obtain that 
    \begin{align}\label{eq:sub_lem_eigenvector_P}
        \braket{\tau|P|\tau}=\sum_{\bfa,\bfb}c^{\ast}_{\bfa}P_{\bfa,\bfb}c_{\bfb}=\frac{1}{2^n}\braket{\tilde{\tau}|\widetilde{P}|\tilde{\tau}}\ge0,
    \end{align}
    where we defined $\ket{\tilde{\tau}}\equiv \sum_{\bfa,\bfc\in \bfZ^n_2}c_{\bfa}\ket{\bfc,\bfa}$, and the last inequality follows by the fact that $\widetilde{P}\ge0$. In conclusion, $P\ge 0$.

    Next, let us assume $d=2$. We propagate the noise $\mathcal{N}_1$ to the left of $\mathcal{N}_{2}$. The resulting form $\mathcal{N}'_{(1)}$ is again the Pauli noise, because the mid-circuit is a Clifford operation. More precisely, there exists a symplectic operation $S\in {\rm Sp}(2n,\bfZ_2)$ such that corresponding noise distribution $\widetilde{p}^{(1)'}:=\widetilde{p}^{(1)}\circ S$, and its matrix representation $\widetilde{P}^{(1)'}=\sigma_S \widetilde{P}^{(1)}\sigma_S^{T}$ for some symmetric operator $\sigma_S\in \mathfrak{S}_{4^n}$. Then we note that $\widetilde{P}^{(1)'}$ is a positive semi-definite matrix representation of another noise distribution $\widetilde{p}^{(1)}\circ S$ ($\widetilde{P}^{(1)'}_{\bfa,\bfb}=\widetilde{p}^{(1)}_{S\bfa+S\bfb}=\widetilde{p}^{(1)}_{S(\bfa+\bfb)}$), and commutes with $\widetilde{P}^{(2)}$, which is again positive semi-definite. Specifically, the commutativity can be shown by the fact that
    \begin{align}\label{eq:stochastic_matrix_commutativity}
        \left(\widetilde{P}^{(1)'}\widetilde{P}^{(2)}\right)_{\bfa,\bfb}=\sum_{\bfc\in \bfZ_2^{2n}}\widetilde{p}^{(1)}_{S(\bfa+\bfc)}\widetilde{p}^{(2)}_{\bfc+\bfb}&=\sum_{\bfc\in \bfZ_2^{2n}}\widetilde{p}^{(1)}_{S(\bfc+\bfb)}\widetilde{p}^{(2)}_{\bfc+\bfa}\nonumber\\&=\left(\widetilde{P}^{(2)}\widetilde{P}^{(1)'}\right)_{\bfa,\bfb}.
  \end{align}
  Conclusively, $\widetilde{P}^{(1)'}\widetilde{P}^{(2)}\ge 0$. 

  Finally, the first line in Eq.~\eqref{eq:stochastic_matrix_commutativity} implies that $\widetilde{P}^{(1)'}\widetilde{P}^{(2)}$ is the matrix representation of $\widetilde{p}^{(1,2)}:=\widetilde{p}^{(1)'}\ast \widetilde{p}^{(2)}$. Inductively, the proof naturally extends to arbitrary $d$, reaching $\widetilde{p}^{(1,2,3,\ldots,d)}$.  
\end{proof}

In addition, if we could estimate the $p_{\mathbf{0}}$ via some fidelity estimation~\cite{flammia2011,huang2020} or verification~\cite{zhu2019,huang2025_sv}, we can guarantee $\widehat{p}\ge 0$ in other way. Since $\sum_{\bfa}\widehat{p}_{\mathbf{a}}=\sum_{\bfa}\sqrt{\widehat{\mu}}_{\bfa}=2^np_{\mathbf{0}}$, we check whether $\sum_{\bfa}\sqrt{\widehat{\mu}}_{\bfa}=2^np_{\mathbf{0}}$ while regarding $\forall \sqrt{\widehat{\mu}}_{\bfa}\ge 0$. If yes, then there is no negative value in $\widehat{p}$. 

Next, we show the required time complexity for WHT. 
\begin{proposition}\label{thm:time_solving_noise_equation}
Suppose $\mu$ has a support $A\subset \bfZ^n_2$ whose generated subspace is known as $V={\rm span
}\left\{\bfv_1,\bfv_2,\ldots,\bfv_{{\rm dim}(V)}\right\}$. The time complexity to calculate
$\frac{1}{2^n}H\sqrt{H\mu}$ is $\mathcal{O}(|V|\log (|V|))$.
\end{proposition}

\begin{proof}
    We look into the transformation. Given $\bfb\in \bfZ^n_2$,

    \begin{align}
        (H\mu)_{\bfb}=\sum_{\bfa\in A}(-1)^{\bfb\cdot \bfa}\mu_{\bfa}.
    \end{align}
    We decompose $\bfb=\bfb_1+\bfb_2$, where $\bfb_1\in V,\;\bfb_2\in V^{\perp}$. Then we note that $(H\mu)_\bfb$ depends only on the $\bfb_1$, not $\bfb_2$. Hence let us say $(H\mu)_\bfb=(H\mu)_{\bfb_1}$. Taking the square [or ($p\in \mathbb{N}$)-th] root to them requires $\mathcal{O}(|V|)$-time. Next, the final value is for $\bfu=\bfu_1+\bfu_2\in \bfZ^n_2$ with the same manner of decomposition,
    \begin{align}
        p_\bfu&=\frac{1}{2^n}\sum_{\bfb_1\in V,\bfb_2\in V^{\perp}}(-1)^{\bfu\cdot (\bfb_1+\bfb_2)}\sqrt{(H\mu)_{\bfb_1}}\nonumber\\&=\frac{1}{2^n}\sum_{\bfb_1\in V}(-1)^{\bfu_1\cdot \bfb_1}\sqrt{(H\mu)_{\bfb_1}}\sum_{\bfb_2\in V^{\perp}}(-1)^{\bfu_2\cdot \bfb_2}\nonumber\\&=\frac{1}{2^{n-\log_2(|V^{\perp}|)}}\delta_{\bfu_2,\mathbf{0}}\sum_{\bfb_1\in V}(-1)^{\bfu_1\cdot \bfb_1}\sqrt{(H\mu)_{\bfb_1}}.
    \end{align}
    Therefore, $p$ also has support on $V$. Furthermore, when we operate WHT, we can treat $V$ as a subspace spanned by some of the trivial bases as usual. Hence, it takes $\mathcal{O}(|V|\log(|V|))$-time~\cite{alman2023} for WHT. We easily note that $\mathcal{O}(|V|)$-sized memory is required.
\end{proof}

Since we do a finite number of samples, $\mu'$ is slightly different from $\mu$, and hence is $p'=\frac{1}{2^n}H\sqrt{H\mu'}$ from $p$. We showed the well-approximated behavior of $p'$ following the distance between $\mu$ and $\mu'$. The main statement was presented in Prop.~\ref{main:thm_uniqueness}, and we give a detailed proof here. 

    Suppose $\mu$ has a solution $p$ and $\mu'$ has a solution $p'$.
    We start to prove the statement (ii,iii). We take the upper bound of  $\|p-p'\|_1$ as
    \begin{align}\label{eq:p_collapse}
        \|p-p'\|_1\le \frac{1}{2^n}\|H\|_1\|\sqrt{H\mu}-\sqrt{H\mu'}\|_1.
    \end{align}
    We use the fact that for $a,b\in \mathbb{R}^{2^n}_{\ge 0}$,
    \begin{align}\label{eq:upperbound_distance}
        &\sum_{\bfu}|\sqrt{a_{\bfu}}-\sqrt{b_{\bfu}}|\nonumber\\&=\sum_{\bfu}\frac{|a_{\bfu}-b_{\bfu}|}{\sqrt{a_{\bfu}}+\sqrt{b_{\bfu}}}\nonumber\\&\le \frac{1}{\min_{\bfu\in \bfZ^n_2}\left\{\sqrt{a_{\bfu}}+\sqrt{b_{\bfu}}\right\}}\sum_{\bfu}|a_{\bfu}-b_{\bfu}|
        \nonumber\\&\le \frac{1}{\min_{\bfu\in \bfZ^n_2}\left\{\sqrt{a_{\bfu}}\right\}+\min_{\bfu\in \bfZ^n_2}\left\{\sqrt{b_{\bfu}}\right\}}\sum_{\bfu}|a_{\bfu}-b_{\bfu}|.
    \end{align}
    Then Eq.~\eqref{eq:p_collapse} becomes, 
    \begin{align}\label{eq:l_1_dist_upperbound}
        \|p-p'\|_1&\le \frac{1}{2^{n}\min_{\bfa}\left\{(Hp)_{\bfa}\right\}}\|H\|_1\|H(\mu-\mu')\|_1\nonumber\\&\le \frac{1}{2^{n}\min_{\bfa}\left\{(Hp)_{\bfa}\right\}}\|H\|^2_1\|\mu-\mu'\|_1\nonumber\\&\le \mathcal{O}(2^n)\|\mu-\mu'\|_1,
    \end{align}
    because $\|H\|_1=2^n$ and $\forall \widehat{p}_{\bfa}\ge 1-2\delta\ge \Omega(1)$ given that $\delta<\frac{1}{2}-\Omega(1)$. We stress that the upper bound is very rough. In many cases where $\mu$ and $\mu'$ are very close, we will see that $\mathcal{O}(2^n)$ reduces to a constant.  

    Before describing in detail, we go over the upper bound of the $l_2$-norm. It can be found similarly to Eq.~\eqref{eq:upperbound_distance}.
    \begin{align}
        \|p-p'\|_2&\le \frac{1}{2^n}\|H\|_2\|\sqrt{H\mu}-\sqrt{H\mu'}\|_2\nonumber\\&\le \frac{1}{2^{n}
        \min_{\bfa}\left\{(Hp)_{\bfa}\right\}}\|H\|_2^2\|\mu-\mu'\|_2\nonumber\\&\le \mathcal{O}(\min_{\bfa}\widehat{\mu}_\bfa^{-1})\|\mu-\mu'\|_2,
    \end{align}
     where we use the fact that $\|H\|_2=\sqrt{2^n}$. We note that given $\delta<\frac{1}{2}$, $\mathcal{O}(\min_{\bfa}\widehat{\mu}_\bfa^{-1})=\mathcal{O}(1)$.  Therefore, accurate estimation of $\mu$ also guarantees the estimation of $p$ with a similar accuracy. 

    Next, we will show that such an upper bound of Eq.~\eqref{eq:l_1_dist_upperbound} can be tightened. We define 
    \begin{align}
        P_{\delta}\equiv \left\{p\in \mathbb{R}^{2^n}_{\ge 0}\bigg|p_{\mathbf{0}}\ge 1-\delta,\;\sum_{\bfa}p_{\bfa}=1\right\}.
    \end{align}

    The inverse function theorem~\cite{taylor2020,guillemin2010} implies that since the mapping $f(p)_{\bfu}\equiv\sum_{\bfa}p_{\bfa}p_{\bfa+\bfu}$ is analytic and bijective between $P_{\delta}$ and $f(P_\delta)$ (note that this is also open in $\mathbb{R}^{2^n}$), there exists smooth inverse $f^{-1}:f(U)\rightarrow U$ where $U\subset \mathbb{R}^{2^n}$ is some neighborhood of $P_{\delta}$ and analyticity on $B_{\epsilon_{\rm th}}(\mu)\subset f(U)$ for some $\epsilon_{\rm th}>0$. 

    While we suppose $\|\mu-\mu'\|_1\le \epsilon<\epsilon_{\rm th}$, the above consideration means that via Taylor expansion,
    \begin{align}\label{eq:inverse_p_exact}
       \|p'-p\|_1&=\|\nabla f^{-1}(\mu) \cdot(\mu'-\mu)+\mathcal{O}(\epsilon^2)\|_1\nonumber\\&\le \|\nabla f^{-1}(\mu)\|_{1}\|\mu'-\mu\|_1+\mathcal{O}(\epsilon^2). 
    \end{align}
     We have a well-known fact 
     \begin{align}
         \nabla f^{-1}(\mu)=(\nabla f(f^{-1}(\mu)))^{-1}=(\nabla f(p))^{-1}.
     \end{align}
      Also, we note that $(\nabla f(p))_{\bfa,\bfb}=2p_{\bfa+\bfb}$, hence $\frac{1}{2}(\nabla f(p))=I-\delta K$, where
    \begin{align}
        K_{\bfa,\bfb}=
        \begin{cases}
            1&(\bfa=\bfb)\\
           -\frac{1}{\delta}p_{\bfa+\bfb}&(\bfa\ne \bfb)
        \end{cases}.
    \end{align}
    Furthermore, for $\bfx\in \bfZ^n_2$ on the unit ball with $l_1$-norm, 
    \begin{align}
        \sum_{\bfa\in \bfZ^n}\left|\sum_{\bfb\in \bfZ^n_2}K_{\bfa,\bfb}x_{\bfb}\right|&\le \sum_{\bfa,\bfb\in \bfZ^n_2}\left|K_{\bfa,\bfb}x_{\bfb}\right|\nonumber\\&=\sum_{\bfb\in \bfZ^n_2}\left(1+\frac{1}{\delta}\sum_{\bfa\ne \bfb}p_{\bfa+\bfb}\right)\left|x_{\bfb}\right|\nonumber\\&\le (1+1)\|\bfx\|_1=2.
    \end{align}
    The last line follows by $\sum_{\bfa\ne \bfb}p_{\bfa+\bfb}=1-p_{\mathbf{0}}=\delta$.
    Therefore, $\|K\|_{1}\le 2$. We conclude that along with the fact that $\|A^k\|_{1}\le \|A\|_{1}^k$ ($k\in \mathbb{N}$) and $\delta<\frac{1}{2}$,
    \begin{align}\label{eq:approx_nabla}
        \|\nabla f^{-1}(\mu)\|_{1}&=\|(\nabla f(p))^{-1}\|_{1}\nonumber\\&=\left\|\frac{1}{2}(I-\delta K)^{-1}\right\|_{1}\nonumber
    \end{align}
    \begin{align}
        &=\frac{1}{2}\|I+\delta K+\mathcal{O}((2\delta)^2)\|_{1}\nonumber\\&\le\frac{1}{2}\left(1+2\delta+\mathcal{O}((2\delta)^2)\right).
    \end{align}
    Here, we use $\|I\|_{1}=1$.
    By Eq.~\eqref{eq:inverse_p_exact}, 
    \begin{align}
        \|p-p'\|_1&\le (\frac{1}{2}+\delta)\|\mu-\mu'\|_1+\mathcal{O}(\|\mu-\mu'\|_1^2).
    \end{align}

     The result means that if we obtain the solution, then the $l_1$ distance is asymptotically twice shorter than the distance between sampled distributions. This is quite counterintuitive because if $p$ and $p'$ undergo the same Markov process, then $l_1$-distance shrinks or is equal. In this context, $p$ and $p'$ undergo self-convolution; effective stochastic processes are slightly different from each other. 

We did not specify the scale of $\epsilon_{\rm th}$. However, the inverse has fractional order since $f$ is a quadratic function. Therefore, we expect that higher-order terms follow the geometric series. If so, this Taylor series converges for $\epsilon<\mathcal{O}(1)$, as long as $f^{-1}(B_{\epsilon}(\mu))$ has no singular value. In other words, we expect that the applicable region of $\epsilon$ is still within $\mathcal{O}(1)$. For example, let us calculate the second-order term of Eq.~\eqref{eq:inverse_p_exact}.

 We note that for fixed $\bfb\in A$,
    \begin{align}
        0&=\partial_{\bfa,\bfa'}p_{\bfb}=\partial_{\bfa,\bfa'}(f^{-1}\circ f(p))_\bfb\nonumber\\&=\partial_{\bfa}\left(\sum_{\bfc}\partial^{(f)}_{\bfc}f^{-1}_{\bfb}\partial_{\bfa'}f_{\bfc}\right)\nonumber\\&=\sum_{\bfc,\bfc'}\partial_{\bfc',\bfc}^{(f)}f^{-1}_{\bfb}\partial_{\bfa}f_{\bfc'}\partial_{\bfa'}f_{\bfc}+\sum_{\bfc}\partial^{(f)}_{\bfc}f^{-1}_{\bfb}\partial_{\bfa,\bfa'}f_{\bfc},
    \end{align} where we define $\partial^{(f)}$ as the partial derivative with respect to co-domain variable $f$.
    Now, we regard $\partial_{\bfc',\bfc}f^{-1}_{\bfb}$ as $4^n$-sized vector $(H(f_{\bfb}))_{\bfc,\bfc'}\equiv\partial^{(f)}_{\bfc,\bfc'}f^{-1}_{\bfb}$ Notice that $\partial_{\bfa,\bfa'}f_{\bfc}=2\delta_{\bfc,\bfa+\bfa'}$. 
    Then $\sum_{\bfc}\partial_{\bfc}^{(f)}f^{-1}_{\bfb}\partial_{\bfa,\bfa'}f_{\bfc}=2 \partial_{\bfa+\bfa'}^{(f)}f^{-1}_{\bfb}$ is $4^n$-sized vector denoted as $\nabla_2 f_{\bfb}$ with the argument ($\bfa,\bfa'$). Then we obtain a matrix equation, 

    \begin{align}
        &(\nabla f\otimes \nabla f)H(f_{\bfb})=-\nabla_2 f_{\bfb}\nonumber\\&\Rightarrow H(f_{\bfb})=-(\nabla f^{-1}\otimes \nabla f^{-1})\nabla_2 f_{\bfb}.
    \end{align}
    Therefore, using $\nabla^{(f)}(f^{-1})=\nabla f^{-1}$ and the convention $\partial:=\partial^{(f)}$, we obtain that
    \begin{align}
        (\partial^{(f)}_{\bfc',\bfc}f^{-1})_{\bfb}=-2\sum_{\bfa,\bfa'}(\nabla f^{-1})_{\bfc',\bfa'}(\nabla f^{-1})_{\bfc,\bfa}(\nabla f^{-1})_{\bfa+\bfa',\bfb}.
    \end{align}
    Now, we calculate the second-order part of the Taylor expansion, which is

    \begin{align}
        &\frac{1}{2}\sum_{\bfc,\bfc'}\partial_{\bfc',\bfc}f^{-1}_{\bfb}(\mu_\bfc-\mu'_\bfc)(\mu_{\bfc'}-\mu'_{\bfc'})\nonumber\\&=-\sum_{\bfa,\bfa'}\left(\sum_{\bfc'}(\nabla f^{-1})_{\bfc',\bfa'}(\mu_{\bfc'}-\mu'_{\bfc'})\right)\nonumber\\&\cdot\left(\sum_{\bfc}(\nabla f^{-1})_{\bfc,\bfa}(\mu_\bfc-\mu'_\bfc)\right)(\nabla f^{-1})_{\bfa+\bfa',\bfb}.
    \end{align}

    Therefore, the $l_1$-norm of the above value is bounded as, 

        \begin{align}
            &\frac{1}{2}\sum_{\bfb}\left|\sum_{\bfc,\bfc'}\partial_{\bfc',\bfc}f^{-1}_\bfb(\mu_\bfc-\mu_\bfc)(\mu_\bfc'-\mu_\bfc')\right|\nonumber\\&\le\sum_{\bfa,\bfa'}\left|\sum_{\bfc'}(\nabla f^{-1})_{\bfc',\bfa'}(\mu_{\bfc'}-\mu'_{\bfc'})\right|\nonumber\\&\cdot\left|\sum_{\bfc}(\nabla f^{-1})_{\bfc,\bfa}(\mu_\bfc-\mu'_\bfc)\right|\sum_{\bfb}\left|(\nabla f^{-1})_{\bfa+\bfa',\bfb}\right|\nonumber\\&\le\sum_{\bfa'}\left|\sum_{\bfc'}(\nabla f^{-1})_{\bfc',\bfa'}(\mu_{\bfc'}-\mu'_{\bfc'})\right|\nonumber\\&\cdot\sum_{\bfa}\left|\sum_{\bfc}(\nabla f^{-1})_{\bfc,\bfa}(\mu_\bfc-\mu'_\bfc)\right|\left(1+\frac{\delta}{2}\right)\nonumber
        \end{align}

        \begin{align}
            &=\|(\nabla f^{-1})(\mu-\mu')\|_1^2\left(1+\frac{\delta}{2}\right)\nonumber\\&\le (\frac{1}{2}+\delta)^2 (1+\frac{\delta}{2})\|\mu-\mu'\|_1^2 \le \frac{1}{4}(1+\frac{9}{2}\delta)\|\mu-\mu'\|_1^2.
        \end{align}

    Here, we used the fact that $\sum_{\bfb}\left|(\nabla f^{-1})_{\bfa+\bfa',\bfb}\right|=\frac{1+\delta}{2}+\sum_{\bfb\ne \bfa+\bfa'}\frac{p_{\bfa+\bfa'+\bfb}}{2\delta}+\mathcal{O}(\delta^2)\le 1+\frac{\delta}{2}+\mathcal{O}(\delta^2)$. Hence, the rate between the first-order factor and second-order factor is $\frac{2+9\delta}{4+8\delta}+\mathcal{O}(\delta^2)<1$ given that $\delta\ll \frac{1}{2}$.

In Fig.~\ref{fig:numerical_dirichlet_noise}, we estimated the rate between the distance between two different measurement distributions $\|\mu^{(1)}-\mu^{(2)}\|_1$ and the distance between two following solutions $\|p^{(1)}- p^{(2)}\|_1$. We can see that when the noise rate $\delta$ is small, the average rate converges to $\simeq 0.5+\delta$ as we increase $n$. Hence we expect that Eq.~\eqref{eq:asymptotic_p_dist} holds even if we set $\epsilon_{\rm th}=\mathcal{O}(1)$.

\section{Convolution power method (CPM)}\label{method:approx_solution}

In the previous part, we observed that the computational complexity of finding the $p'$ takes $\mathcal{O}(n2^n)$-time, which is inefficient for large $n$. Ref.~\cite{scheibler2015,li2015,harper2021} suggested another FWHT routine, which employs so-called sub-sampling and peeling decoder, specialized to the sparse support. However, this method assumes randomly chosen support over $\bfZ^n_2$, proving the high success rate. This may not be suitable for our cases where the given signal is clustered over $\mathbf{0}$~\cite{harper2021}. Moreover, as we see from $p=\frac{1}{2^n}H\sqrt{H\mu}$, we need to operate the Walsh-Hadamard transform twice. It means that we have a non-sparse input for the second transform. To overcome such problems, we introduce another method to approximately solve the convoluted noise equation, the so-called convolution power method (CPM). 
To do so, we need the following lemma, 
\begin{lemma}[Young's inequality]\label{lem:Young_ineq}~\cite{bogachev1950}
Given that $1\le k,m,l\le \infty$ such that $\frac{1}{k}+\frac{1}{m}=\frac{1}{l}+1$, $\|p\ast q\|_l\le \|p\|_k\|q\|_m$. 
\end{lemma}

We first denote $\bm{\delta}_p\equiv \vec{1}-p$. Here, $\|\bm{\delta}_p\|_2\le\|\bm{\delta}_p\|_1=2(1-p_{\mathbf{0}})=2\delta$. Similarly, we can obtain that $\|\bm{\delta}_\mu\|_{1(\;\rm or\;2)}\le 4\delta+\mathcal{O}(\delta^2)$. As a result, we shall omit the subscript and just denote the norms as $\|\cdot\|$. Before proceeding, we need the following formula. 
\begin{lemma} Let $w\in \mathbb{N}$, and suppose $\|\bm{\delta}_p\|_1=2\delta<1$. Then,
    $(\vec{1}-\bm{\delta}_p)^{\ast -1}=\sum_{l=0}^{\infty}\bm{\delta}_p^{\ast l}=\sum_{l=0}^{w-1}\bm{\delta}_p^{\ast l}+\mathcal{O}\left(\frac{(2\delta)^w}{1-2\delta}\right) $.
\end{lemma}

\begin{proof}
    It is necessary to show that such a geometric series is well-defined. 
    Given $\bfx\in \mathbb{R}^{2^n}$. We denote $|\bfx|$ as the maximal absolute value of elements in $\bfx$. 
    $(\vec{1}-\bm{\delta})^{\ast -1}\ast (\vec{1}-\bm{\delta}_p)=\vec{1}\;\Rightarrow \widehat{(\vec{1}-\bm{\delta})^{\ast -1}}=\widehat{ (\vec{1}-\bm{\delta}_p)}^{-1}=(1-\widehat{\bm{\delta}_p})^{-1}=\sum_{l=0}^{\infty}\widehat{\bm{\delta}_p}^l=\sum_{l=0}^{w-1}\widehat{\bm{\delta}_p}^l+\bfa$, where $|\bfa|=\mathcal{O}((2\delta)^w)$ and the last two equations hold since $|\widehat{\bm{\delta}_p}|=|\sum_{\bfb}\bm{\delta}_{p\bfb}(-1)^{\bfa\cdot \bfb}|\le \sum_{\bfb}|\bm{\delta}_{p\bfb}|=2\delta<1$. Therefore, after the inverse transform, the equation becomes $(\vec{1}-\bm{\delta}_p)^{\ast -1}=\sum_{l=0}^{\infty}\bm{\delta}_p^{\ast l}=\sum_{l=0}^{w}\bm{\delta}_p^{\ast l}+\bfa'$, where $\bfa'=\sum_{l=w}^{\infty}\bm{\delta}_p^{\ast l}=\bm{\delta}_p^{\ast w}\ast\sum_{l=0}^{\infty}\bm{\delta}_p^{\ast l}$. We proved that the series is well-defined. Next, by Young's inequality, $\|\bfa'\|_1\le \left\|\bm{\delta}_p^{\ast w}\right\|_1\left\|\sum_{l=0}^{\infty}\bm{\delta}_p^{\ast l}\right\|_1\le\left\|\bm{\delta}_p\right\|^ w_1\sum_{l=0}^{\infty}\|\bm{\delta}_p\|_1^ l=\frac{\left\|\bm{\delta}_p\right\|^w_1}{1-\|\bm{\delta}_p\|_1}\le\frac{(2\delta)^w}{1-2\delta}$. 
\end{proof}
Now, we suppose $\delta<\frac{1}{4}$. Then we can rewrite the equation $p\ast p=\mu$
 to $(\vec{1}-\bm{\delta}_p)\ast p=\mu$ and hence given $w\in \mathbb{N}$,
\begin{align}\label{eq:inverse_p}
    p=(\vec{1}-\bm{\delta}_p)^{\ast-1}\ast\mu&=\sum_{l=0}^{w-1}\bm{\delta}_p^{\ast l}\ast \mu+\mathcal{O}\left(\frac{(2\delta)^w}{1-2\delta}\right) \nonumber\\&=\sum_{l=0}^{w-1}\bm{\delta}_p^{\ast l}\ast \mu+\mathcal{O}\left((2\delta)^w\right). 
\end{align}

We can rewrite Eq.~\eqref{eq:inverse_p} as, 

\begin{align}
    &p=\sum_{l=0}^{w-1}\sum_{k=0}^{l}\comb{l}{k}(-1)^{k}p^{\ast k}\ast\mu+\mathcal{O}((2\delta)^w)\nonumber\\
    &=\sum_{l=0}^{w-1}\left\{\sum_{k=0}^{\lfloor\frac{l}{2}\rfloor}\comb{l}{2k}p^{\ast 2k}\ast\mu-\sum_{k=0}^{\lfloor\frac{l}{2}\rfloor}\comb{l}{2k+1}p^{\ast (2k+1)}\ast\mu\right\}\nonumber\\&+\mathcal{O}((2\delta)^w)\nonumber\\&=\sum_{l=0}^{w-1}\left\{\sum_{k=0}^{\lfloor\frac{l}{2}\rfloor}\comb{l}{2k} \mu^{\ast k+1}-\sum_{k=0}^{\lfloor\frac{l}{2}\rfloor}\comb{l}{2k+1}p\ast \mu^{\ast k+1}\right\}\nonumber\\&+\mathcal{O}((2\delta)^w),
\end{align}
where the last equality used the convolution properties above and $p\ast p=\mu$. The rule that $\comb{0}{0}=1,\;\comb{a}{b}=0$ if $a<b$ has been applied. Therefore, we obtain that 
{\small
\begin{align}
    &p+\sum_{l=0}^{w-1}\sum_{k=0}^{\lfloor\frac{l}{2}\rfloor}\comb{l}{2k+1}\mu^{\ast k+1}\ast p=\sum_{l=0}^{w-1}\sum_{k=0}^{\lfloor\frac{l}{2}\rfloor}\comb{l}{2k}\mu^{\ast k+1}+\mathcal{O}((2\delta)^w)\nonumber\\
   &\Rightarrow \left\{\vec{1}+\sum_{l=0}^{w-1}\sum_{k=0}^{\lfloor\frac{l}{2}\rfloor}\comb{l}{2k+1}\mu^{\ast k+1}\right\}\ast p\nonumber\\&=\sum_{l=0}^{w-1}\sum_{k=0}^{\lfloor\frac{l}{2}\rfloor}\comb{l}{2k}\mu^{\ast k+1}+\mathcal{O}((2\delta)^w)\nonumber\\&\Rightarrow p=\left\{\vec{1}+\sum_{l=0}^{w-1}\sum_{k=0}^{\lfloor\frac{l}{2}\rfloor}\comb{l}{2k+1}\mu^{\ast k+1}\right\}^{\ast-1}\nonumber\\&\ast\left(\sum_{l=0}^{w-1}\sum_{k=0}^{\lfloor\frac{l}{2}\rfloor}\comb{l}{2k}\mu^{\ast k+1}+\mathcal{O}((2\delta)^w)\right).\nonumber\\&\Rightarrow
   p=\left\{\vec{1}+\sum_{l=0}^{w-1}\sum_{k=0}^{\lfloor\frac{l}{2}\rfloor}\sum_{m=0}^{k+1}\comb{l}{2k+1}\cdot\comb{k+1}{m}(-1)^{m}\bm{\delta}_\mu^{\ast m}\right\}^{\ast-1}\nonumber\\&\ast \left(\sum_{l=0}^{w-1}\sum_{k=0}^{\lfloor\frac{l}{2}\rfloor}\comb{l}{2k}\mu^{\ast k+1}+\mathcal{O}((2\delta)^w)\right)=\nonumber
\end{align}

\begin{align}\label{eq_noise_solution_another}
     &\left\{\vec{1}+\sum_{m=1}^{\lfloor\frac{w-1}{2}\rfloor+1}(-1)^m\frac{\sum_{l=2\lfloor\frac{m}{2}\rfloor}^{w-1}\sum_{k=m-1}^{\lfloor\frac{l}{2}\rfloor}\comb{l}{2k+1}\cdot\comb{k+1}{m}\bm{\delta}_\mu^{\ast m}}{1+\sum_{l=0}^{w-1}\sum_{k=0}^{\lfloor\frac{l}{2}\rfloor}\comb{l}{2k+1}}\right\}^{ \ast-1}\nonumber\\&\ast\left(\frac{\sum_{l=0}^{w-1}\sum_{k=0}^{\lfloor\frac{l}{2}\rfloor}\comb{l}{2k}\mu^{\ast k+1}}{1+\sum_{l=0}^{w-1}\sum_{k=0}^{\lfloor\frac{l}{2}\rfloor}\comb{l}{2k+1}}+\mathcal{O}\left(\delta^w\right)\right),
\end{align}
}where we used the summation indices equivalence,
{\small
\begin{align}\label{eq:summation_equivalence}
	\sum_{l=0}^{w-1}\sum_{k=0}^{\lfloor\frac{l}{2}\rfloor}\sum_{m=0}^{k+1}=\sum_{m=0}^{w-1}\sum_{m=0}^{\lfloor\frac{l}{2}\rfloor+1}\sum_{k=m-1}^{\lfloor\frac{l}{2}\rfloor}=\sum_{m=0}^{\lfloor\frac{w-1}{2}\rfloor+1}\sum_{l=2\lfloor\frac{m}{2}\rfloor}^{w-1}\sum_{k=m-1}^{\lfloor\frac{l}{2}\rfloor},
\end{align} 
}
and, 
\begin{align}
	1+\sum_{l=0}^{w-1}\sum_{k=0}^{\lfloor\frac{l}{2}\rfloor}\comb{l}{2k+1}=1+\sum_{l=0}^{w-1}\left(\sum_{j=0,{\rm odd}}^{l}\comb{l}{j}\right)=2^{w-1}.
\end{align}

Moreover, we note that for fixed $m\in \mathbb{N}$,
\begin{align}\label{eq:summationterm_inverse}
    &\left\|\sum_{m=1}^{\lfloor\frac{w-1}{2}\rfloor+1}(-1)^m\frac{\sum_{l=2\lfloor\frac{m}{2}\rfloor}^{w-1}\sum_{k=m-1}^{\lfloor\frac{l}{2}\rfloor}\comb{l}{2k+1}\cdot\comb{k+1}{m}\bm{\delta}_\mu^{\ast m}}{1+\sum_{l=0}^{w-1}\sum_{k=0}^{\lfloor\frac{l}{2}\rfloor}\comb{l}{2k+1}}\right\|\nonumber\\&=\left\|\sum_{m=1}^{\lfloor\frac{w-1}{2}\rfloor+1}(-1)^m d(w,m)\bm{\delta}_\mu^{\ast m}\right\|\nonumber\\
    &\le \sum_{m=1}^{\lfloor\frac{w-1}{2}\rfloor+1}\left(\max_{m\in [\lfloor\frac{w-1}{2}\rfloor+1]}\left\{d(w,m)^{1/m}\right\}\|\bm{\delta}_\mu\|_1\right)^m,
\end{align}
where we define the non-negative function $d(w,m)\equiv \frac{\sum_{l=2\lfloor\frac{m}{2}\rfloor}^{w-1}\sum_{k=m-1}^{\lfloor\frac{l}{2}\rfloor}\comb{l}{2k+1}\cdot\comb{k+1}{m}}{1+\sum_{l=0}^{w-1}\sum_{k=0}^{\lfloor\frac{l}{2}\rfloor}\comb{l}{2k+1}}$ for brevity. We also define $d_{\rm max}(w)\equiv \max_{m\in [\lfloor \frac{w-1}{2}\rfloor +1]}\left\{d(w,m)^{1/m}\right\}$. If $\delta< \frac{1}{12d_{\rm max}(w)}$ while $\|\bm{\delta}_\mu\|_1<4\delta<\frac{1}{3d_{\rm max}(w)}$, then the last geometric series of Eq.~\eqref{eq:summationterm_inverse} becomes lower than $\frac{1}{2}$. The matrix norm of the inverse part in Eq.~\eqref{eq_noise_solution_another} is then upper bounded by $\mathcal{O}(1)$. The $d(w,m)$ can be efficiently calculated and hence is $d_{\rm max}(w)$. Fig.~\ref{fig:d_degree_approx} (a) shows that $d_{\rm max}(w)$ scales as $\frac{w}{4}$. That is, the right hand side of Eq.~\eqref{eq:summationterm_inverse} is at most $\frac{w\delta}{1-w\delta}\le\frac{w\delta}{1-\frac{1}{3}}=\frac{3w\delta}{2}$ and $\delta<\frac{1}{3w}$ is a sufficient condition for the infinite series expansion. 

We fix another number $s\in \mathbb{N}$. The infinite series expansion of the inverse can be truncated up to $(w+s)$ number of terms so that the bias is $\mathcal{O}(\left(\frac{3w\delta}{2}\right)^{w+s})$.

We also note that $\left\|\frac{\sum_{l=0}^{w-1}\sum_{k=0}^{\lfloor\frac{l}{2}\rfloor}\comb{l}{2k}\mu^{\ast k+1}}{\sum_{l=0}^{w-1}\sum_{k=0}^{\lfloor\frac{l}{2}\rfloor}\comb{l}{2k+1}}\right\|\le \mathcal{O}(1)$. In conclusion, Eq.~\eqref{eq_noise_solution_another} is rewritten into the final form,  

{\small
\begin{align}\label{eq_noise_solution_another_final}
    &p=\sum_{t=0}^{w+s-1}(-1)^t\times\nonumber\\&\left\{\sum_{m=1}^{\lfloor\frac{w-1}{2}\rfloor+1}(-1)^m\frac{\sum_{l=2\lfloor\frac{m}{2}\rfloor}^{w-1}\sum_{k=m-1}^{\lfloor\frac{l}{2}\rfloor}\comb{l}{2k+1}\cdot\comb{k+1}{m}\bm{\delta}_\mu^{\ast m}}{1+\sum_{l=0}^{w-1}\sum_{k=0}^{\lfloor\frac{l}{2}\rfloor}\comb{l}{2k+1}}\right\}^{\ast t}\nonumber\\&\ast\frac{\sum_{l=0}^{w-1}\sum_{k=0}^{\lfloor\frac{l}{2}\rfloor}\comb{l}{2k}\mu^{\ast k}}{1+\sum_{l=0}^{w-1}\sum_{k=0}^{\lfloor\frac{l}{2}\rfloor}\comb{l}{2k+1}}+\mathcal{O}\left(\left(\frac{3w\delta}{2}\right)^{w+s}+\left(\delta^w\right)\right).
\end{align}}

The approximated term should be rewritten as $\bm{\delta}_\mu=1-\mu$ so that the following expression can be found by sampling from $\mu$. 
\begin{figure}[t]
    
    \centering

    \includegraphics[width=\columnwidth]{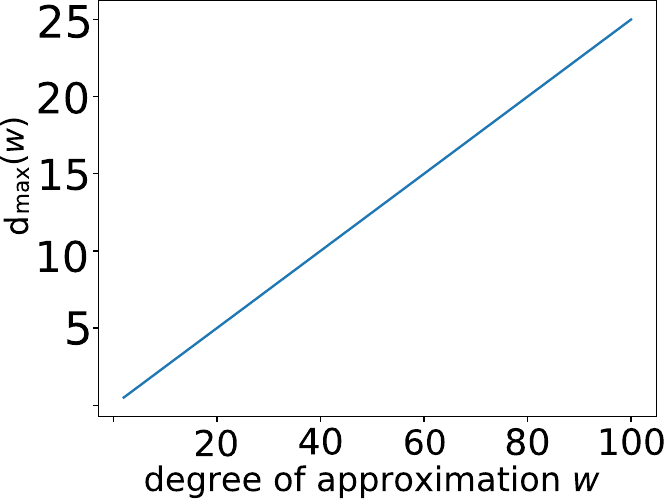}

    \caption{Scaling of $d_{\rm max}(w)\equiv \max_{m\in [\lfloor\frac{w-1}{2}\rfloor+1]}\left\{d(w,m)^{1/m}\right\}$ in Eq.~\eqref{eq:summationterm_inverse}. The maximum happens when $m=1$
}\label{fig:d_degree_approx}
    
\end{figure}

Let us give some examples. If we choose $(w,s)=(2,0)$, we have the condition $\delta<0.166\ldots$. Then, Eq.~\eqref{eq_noise_solution_another_final} reduces to 
\begin{align}
    p&=\left(\vec{1}-\frac{1}{2}\bm{\delta}_\mu\right)^{-1}\ast\left(\mu+\mathcal{O}((2\delta)^2)\right)\nonumber\\&=\left(\frac{3}{2}\vec{1}-\frac{1}{2}\mu\right)\ast\mu+\mathcal{O}((3\delta)^2+(2\delta)^2)\nonumber\\&=\frac{1}{2}\left(3\mu^{\ast 1}-\mu^{\ast 2}\right)+\mathcal{O}((3\delta)^2). 
\end{align}

If we choose $(w,s)=(3,0)$ ($\delta<0.111\ldots$), then we have
\begin{align}
    &p\nonumber\\&=\left(\vec{1}-\frac{3}{4}\bm{\delta}_\mu\right)^{-1}\ast\frac{1}{4}\left(3\mu+\mu\ast \mu\right)+\mathcal{O}(\delta^3)\nonumber\\&=\left(\vec{1}+\frac{3}{4}\bm{\delta}_\mu+\frac{9}{16}\bm{\delta}_\mu^2\right)\ast\frac{1}{4}\left(3\mu+\mu\ast \mu\right)+\mathcal{O}((4.5\delta)^3+\delta^3)\nonumber\\&=\left(\frac{37}{16}\vec{1}-\frac{15}{8}\mu+\frac{9}{16}\mu^{\ast 2}\right)\ast\frac{1}{4}\left(3\mu+\mu\ast \mu\right)+\mathcal{O}((4.5\delta)^3)\nonumber\\&=\frac{111}{64}\mu^{\ast 1}-\frac{53}{64}\mu^{\ast 2}-\frac{3}{64}\mu^{\ast 3}+\frac{9}{64}\mu^{\ast 4}+\mathcal{O}((4.5\delta)^3).
\end{align}
We can see that the condition of $\delta$ is relaxed to be $\delta<\frac{1}{4}$ for $w=2$, and $\delta<\frac{1}{6}$ for $w=3$.
In the same manner, 
{\small
\begin{align}
    p&=1.65527344\mu^{\ast 1}-0.237304688\mu^{\ast 2}-1.18847656\mu^{\ast 3}\nonumber\\&+0.83307812\mu^{\ast 4}+0.0517578125\mu^{\ast 5}-0.0947265625\mu^{\ast 6}\nonumber\\&-0.0185546875\mu^{\ast 7}-0.0009756525\mu^{\ast8}+\mathcal{O}((6\delta)^4)\nonumber\\&({\rm if}\;\delta<0.8333\ldots)\nonumber\\&
    =1.37901783\mu^{\ast 1}+1.03779793\mu^{\ast 2}-2.91566753\mu^{\ast 3}+\nonumber\\&0.900688171\mu^{\ast 4}+1.32462502\mu^{\ast 5}-0.464344025\mu^{\ast 6}\nonumber\\&-0.370359421\mu^{\ast 7}+0.041007996\mu^{\ast 8}+0.059091568\mu^{\ast 9}\nonumber\\&+0.010728836\mu^{\ast10}+0.00059604\mu^{\ast 11}+\mathcal{O}((7.5\delta)^5)\nonumber\\&({\rm if}\;\delta<0.0666\ldots).
\end{align}
}
The $\mathcal{O}((w\delta)^w)$ bias could be a serious factor in the inaccuracy of the estimation. In this case, we can increase the $s$. For instance, when $w=5$ and $s=1$ (or $s=2$), 
{\small
\begin{align}\label{eq:approx_much_higher}
    p&
    =1.60532922\mu^{\ast 1}+0.736049414\mu^{\ast 2}-3.75050509\mu^{\ast 3}\nonumber\\&+2.34237552\mu^{\ast 4}+1.59005195\mu^{\ast 5}-1.54467821\mu^{\ast 6}\nonumber\\&-0.470942257\mu^{\ast 7}+0.415027142\mu^{\ast 8}+0.142522156\mu^{\ast 9}\nonumber\\&-0.0376999378\mu^{\ast10}-0.02361834\mu^{\ast 11}-0.00372529\mu^{\ast 12}\nonumber\\&-0.00018626
    \mu^{\ast 13}+\mathcal{O}((7.5\delta)^6+\left(\delta^5\right))\nonumber\\&
    =1.81749614\mu^{\ast 1}+0.3311715566\mu^{\ast 2}-4.41529479\mu^{\ast 3}\nonumber\\&+4.31002729\mu^{\ast 4}+1.19872184\mu^{\ast 5}-3.1739106\mu^{\ast 6}\nonumber\\&+0.0270242803\mu^{\ast 7}+1.16613880\mu^{\ast 8}+0.021392014\mu^{\ast 9}\nonumber\\&-0.254116021\mu^{\ast 10}-0.043117907\mu^{\ast 11}+0.023052337\mu^{\ast 12}\nonumber\\&+0.009534415
    \mu^{\ast 13}+0.0012805685
    \mu^{\ast 14}\nonumber\\&+0.0000582077
    \mu^{\ast 15}+\mathcal{O}((7.5\delta)^7+\left(\delta^5\right))\;({\rm if}\;\delta<0.0666\ldots).
\end{align}
}

We can do this similarly for higher-order bias. For $(w,s)$-approximation, the calculation encounters at most $\left(\lfloor\frac{w-1}{2}\rfloor+(\lfloor\frac{w-1}{2}\rfloor+1)(w+s-1)\right)$th-iterative convolution, where $(\lfloor\frac{w-1}{2}\rfloor+1)(w+s-1)$ comes from the inverse in Eq.~\eqref{eq_noise_solution_another}. Therefore, we obtain the general expression, 
\begin{align}\label{eq:supp_approx_series}
    p&=\underbrace{\sum_{j=1}^{\lfloor\frac{w+1}{2}\rfloor+(\lfloor\frac{w-1}{2}\rfloor+1)(w+s-1)}c_j \mu^{\ast j}}_{p_{\rm approx}^{(w,s)}}\nonumber\\&+\mathcal{O}\left(\left(\frac{3w\delta}{2}\right)^{w+s}+\left(\delta^w\right)\right),\;{\rm where}\;\forall c_j\in \mathbb{R}.
\end{align}
We see that 
\begin{align}
	1&=\sum_{\bfa\in \bfZ^n_2}\lim_{\delta\rightarrow 0}p_\bfa=\sum_{j=0}^{\lfloor\frac{w+1}{2}\rfloor+(\lfloor\frac{w-1}{2}\rfloor+1)(w+s-1)}c_j \sum_{\bfa\in \bfZ^n_2}\mu^{\ast j}_\bfa\nonumber\\&=\sum_{j=0}^{\lfloor\frac{w+1}{2}\rfloor+(\lfloor\frac{w-1}{2}\rfloor+1)(w+s-1)}c_j.
\end{align}
Therefore, the coefficients of the approximate solution are quasi-probabilistic. Indeed, we can observe that all the above examples of low-degree approximations are quasi-probabilistic.

\section{Projecting quasi-probability onto probability simplex~\cite{wang2013}}\label{app:projecting_quasi_simplex_l2}

Here, we explain the algorithm to project the quasi-probability to the simplex. We suppose a quasi-probability $p'$ of support $|A|\in \bfZ^n_2$ is obtained. Readers refer to Ref.~\cite{wang2013} for its proof. The algorithm is as follows:
\begin{enumerate}
    \item We sort $p'$ into descending order that is, $p'\rightarrow p'^{\downarrow}$.
    \item We find 
    \begin{align}
         \eta=\max\left\{1\le j\le 2^n:p'^{\downarrow}_j+\frac{1}{j}\left(1-\sum_{
    i=1
    }^{j}p'^{\downarrow}_i\right)>0\right\}
    \end{align}
    
    \item We set $\lambda=\frac{1}{\eta}\left(1-\sum_{i=1}^{\eta}p'^{\downarrow}_i\right)$
    \item 
    Output $p^{(+)}$ as $p^{(+)}_{\bfa}=\max\left\{p'_{\bfa}+\lambda,0\right\}$.
\end{enumerate}
The sorting algorithm Step.~$1$ needs $\mathcal{O}(|A|\log(|A|))$-time and memory costs. Next, in the step.~$2$, we first check $j$'s up to when $p'^{\downarrow}_j\ge 0$. Ignoring the interval of $p'^{\downarrow }_j=0$, we may jump to the negative region. Because, when $p^{'\downarrow}_j=0$, $\frac{1}{j}\left(1-\sum_{i=1}^{k}p^{'\downarrow}\right)\le 0$. Therefore, $\eta$ can be found in $\mathcal{O}(|A|)$-time. The $\lambda$ and output can also be computed in $\mathcal{O}(|A|)$-time because the zero interval gets the same value $\lambda$.  
Conclusively, the total time and memory complexity is $\mathcal{O}(|A|\log(|A|))$. Since the probability simplex is convex, the projection merely reduces the distance from $p$ that is, $\|p^{(+)}-p\|_2\le\|p'-p\|_2$~\cite{borwein2006,flammia2020}.

\section{Complexity analysis of CPM}\label{method:sampling_complexity}

In this section, we prove Prop.~\ref{main:prop_sparsity_general}.

    We first prove (i). For convenience, we define for $m\in \mathbb{N}$, 
    \begin{align}
        mA\equiv \left\{\bfa|\bfa'=\bfa_1+\bfa_2+\ldots+\bfa_m\;{\rm for}\; \bfa,\ldots,\bfa_m\in A\right\}.
    \end{align}
    We suppose $\sum_{\bfa\in A}\mu_{\bfa}\ge 1-\epsilon$, and note that $\mu^{\ast j+1}=\sum_{\substack{\bfa_1,\bfa_2,\ldots,\bfa_{j+1}\in \bfZ^n_2\\\bfa_1+\bfa_2+\ldots+\bfa_{j+1}=\bfu}}\left(\prod_{i=1}^{j+1}\mu_{\bfa_i}\right)$. Then we obtain 
{\small
    \begin{align}
        &\sum_{\bfu\notin (j+1)A}\mu^{\ast j+1}_{\bfu}=\sum_{\bfu\notin (j+1)A}\left(\sum_{\substack{\bfa_1,\bfa_2,\ldots,\bfa_{j+1}\in \bfZ^n_2\\ \bfa_1+\bfa_2+\ldots+\bfa_{j+1}=\bfu,\bfa_{j+1}\in A}}\left(\prod_{i=1}^{j+1}\mu_{\bfa_i}\right)\right.\nonumber\\&\left.+\sum_{\substack{\bfa_1,\bfa_2,\ldots,\bfa_{j+1}\in \bfZ^n_2\\\bfa_1+\bfa_2+\ldots+\bfa_{j+1}=\bfu,\bfa_{j+1}\notin A}}\left(\prod_{i=1}^{j+1}\mu_{\bfa_i}\right)\right)\nonumber
    \end{align}
}
\begin{align}\label{eq:supp_CPM_expansion_coefficient}
     &\le \sum_{\bfu\notin (j+1)A}\left(\sum_{\substack{\bfa_1,\bfa_2,\ldots,\bfa_{j+1}\in \bfZ^n_2,\bfa_{j+1}\in A \\\bfa_1+\bfa_2+\ldots+\bfa_{j+1}=\bfu}}\left(\prod_{i=1}^{j+1}\mu_{\bfa_i}\right)\right)\nonumber\\&+\sum_{\bfu\in \bfZ^n_2}\left(\sum_{\substack{\bfa_1,\bfa_2,\ldots,\bfa_{j+1}\in \bfZ^n_2,\bfa_{j+1}\notin A\\\bfa_1+\bfa_2+\ldots+\bfa_{j+1}=\bfu}}\left(\prod_{i=1}^{j+1}\mu_{\bfa_i}\right)\right) \nonumber\\&\le \sum_{\bfu\notin jA}\left(\sum_{\substack{\bfa_1,\bfa_2,\ldots,\bfa_{j+1}\in \bfZ^n_2,\bfa_{j+1}\in A \\\bfa_1+\bfa_2+\ldots+\bfa_{j}=\bfu}}\left(\prod_{i=1}^{j}\mu_{\bfa_i}\right)\cdot \mu_{\bfa_{j+1}}\right)+\epsilon\nonumber\\&\le \sum_{\bfu\notin jA}\mu^{\ast j}_{\bfu}+\epsilon,
    \end{align}

    where the third inequality is derived by that $\bfa_{j+1}+\bfu\notin jA$ for $\bfa_{j+1}\in A$ otherwise $\bfu\in (j+1)A$. By induction, $\sum_{\bfu\in A'}\mu^{\ast j}_{\bfu}\ge 1-j\epsilon$ for some $A'\subset \bfZ^n_2$ with size $\mathcal{O}(|A|^{j})$.  

    Next, we prove (ii). Suppose $\mu'$ is an empirical distribution of a desired one $\mu$ after $N$ samples. By (i), $\sum_{\bfa\in jA}\mu^{\ast j}_{\bfa}\ge 1-j\epsilon$. Furthermore, 
    \begin{align}
        &\sum_{\bfa\in \bfZ^n_2}\left|\mu^{\ast j}_{\bfa}-\mu^{\ast j'}_{\bfa}\right|\nonumber\\&=\sum_{\bfa\in jA}\left|\mu^{\ast j}_{\bfa}-\mu^{\ast j'}_{\bfa}\right|+\sum_{\bfa\notin jA}\left|\mu^{\ast j}_{\bfa}-\mu^{\ast j'}_{\bfa}\right|\nonumber\\&\le \sum_{\bfa\in jA}\left|\mu^{\ast j}_{\bfa}-\mu^{\ast j'}_{\bfa}\right|+\sum_{\bfa\notin jA}\mu^{\ast j}_{\bfa}+\sum_{\bfa\notin jA}\mu^{\ast j'}_{\bfa}.
    \end{align}

    Now, we take large $N\ge M\cdot\mathcal{O}(\frac{|A|^{j}}{\epsilon^{'2}})$~\cite{Jerrum:1986random,berend2012} ($M\gg (1-j\epsilon')^{-1}$) such that $\left|\frac{N_{jA}}{N}-\sum_{\bfa\in jA}\mu^{\ast j}_{\bfa}\right|\le\epsilon'$ with failure probability $\delta_f\le \mathcal{O}\left(e^{-|A|^{j}}\right)$, where $N_{jA}\le N$ be the number of samples belonging to $jA$. Then we obtain that 
    \begin{align}\label{eq:ineq:total_to_portion}
        \sum_{\bfa\in \bfZ^n_2}\left|\mu^{\ast j}_{\bfa}-\mu^{\ast j'}_{\bfa}\right|&\le \sum_{\bfa\in jA}\left|\mu^{\ast j}_{\bfa}-\mu^{\ast j'}_{\bfa}\right|+2j\epsilon+\epsilon'.
    \end{align}
    Now, let $x_{\bfa}$ be the number of samples with the outcome $\bfa$ ($\mu^{\ast j'}_{\bfa}=\frac{x_{\bfa}}{N}$). We also define the conditional probability $\mu^{\star j}_{\bfa}=\frac{\mu^{\ast j}_{\bfa}}{\sum_{\bfa'\in jA}\mu^{\ast j}_{\bfa'}}\;(\bfa\in jA)$. With a failure probability $\delta_f'=\mathcal{O}(e^{-M(1-j\epsilon')})$ and $\epsilon'\ll 1/j$, we note $N_{jA}\ge(1-j\epsilon')N\ge M(1-j\epsilon')\cdot\mathcal{O}(\frac{|A|^{j}}{\epsilon^{'2}})$ so that~\cite{berend2012,devroye1983}
    \begin{align}
        \sum_{\bfa\in jA}\left|\frac{x_{\bfa}}{N_{jA}}-\mu^{\star j}_{\bfa}\right|\le \epsilon'. 
    \end{align}
    We can rewrite the above equation as ($\exists\eta\in [-1,1]$), 
    \begin{align}
        &\sum_{\bfa\in jA}\left|\frac{x_{\bfa}}{N}-\frac{N_{jA}}{N}\mu^{\star j}_{\bfa}\right|\le \frac{N_{jA}}{N}\epsilon'.\nonumber\\&\Rightarrow \sum_{\bfa\in jA}\left|\frac{x_\bfa}{N}-\left(\sum_{\bfa\in jA}\mu^{\ast j}_{\bfa}+\eta\epsilon'\right)\mu^{\star j}_{\bfa}\right|\nonumber\\&\le\left(\sum_{\bfa\in jA}\mu^{\ast j}_{\bfa}+\eta\epsilon'\right)\epsilon',\nonumber\\
        &\Rightarrow \left|\sum_{\bfa\in jA}\left|\mu^{\ast j'}_{\bfa}-\mu^{\ast j}_{\bfa}\right|-|\eta|\epsilon'\sum_{\bfa\in jA}\mu^{\star j}_{\bfa}\right|\le\epsilon'+\epsilon^{'2}.
    \end{align}
    Therefore,
    \begin{align}
        &\sum_{\bfa\in jA}\left|\mu^{\ast j'}_{\bfa}-\mu^{\ast j}_{\bfa}\right|\le \left(1+|\eta|\sum_{\bfa\in jA}\mu^{\star j}_{\bfa}\right)(\epsilon'+\epsilon^{'2})\nonumber\\&\le 2\epsilon'+\mathcal{O}(\epsilon^{'2}). 
    \end{align}
    In conclusion, by Eq.~\eqref{eq:ineq:total_to_portion}, $\sum_{\bfa\in \bfZ^n_2}|\mu^{\ast j}_{\bfa}-\mu^{\ast j'}_{\bfa}|\le 3\epsilon'+2j\epsilon+\mathcal{O}(\epsilon'^2)$ with failure probability $1-(1-\delta_f)(1-\delta_f')=\mathcal{O}(e^{-M})$.

    Lastly, we prove (iii).
    Suppose $\left\{q_{P}\right\}_{P\in\mathcal{P}_n}$ is the quasi-probability with which we sample $P$ from $\frac{|q_P|}{\sum_{P}|q_P|}$ ($P$ may be the collection of Paulis acting on each noisy multi-qubit gate). We consider a generalized case where we enact the recovery Pauli operation $P$ to explain probabilistic error cancellation of a noisy convolution circuit. Here, we assume the probability distribution $\frac{|q_P|}{\sum_{P}|q_P|}$ is same for each sampling copy to estimate $\mu$. A more general case will be discussed in  App.~\ref{method:detection_noisy_circuit}. We denote the sampling copy as $(P, \bfu)$. It is the sampled Pauli operator $P$ for mitigation and the measurement outcome $\bfu$ for the whole circuit. Here, $\bfu$ follows $\bar{\mu}$, which is the measurement distribution after the noise and recovery. When the circuit is pure, $\sum|q_P|=1$ and only $(I,\bfu)$'s are sampled following $\mu \;(\mu=\bar{\mu})$. We utilize the proof technique of the Hoeffding inequality in Hilbert space~\cite{devroye2013,pinelis1994} for our problem. Readers may refer to Ref.~\cite{devroye2013,pinelis1994,pinelis1986,yurinskiui1976} for its mathematical details. 

    We first consider when $j=1$. The case of $j>1$ follows a similar logic. Define the indicator function $\delta$ such that $ \delta(P,\bfv)_{\bfu}\equiv\delta_{\bfv,\bfu}\sgn(q_{P})\sum_{P}|q_{P}|$. Suppose we have the $N$ number of sampling copies $\left\{(P_1,\bfv_1),(P_2,\bfv_2),\ldots,(P_N,\bfv_N)\right\}$. Let us define the set of vectors as $X_i\equiv \delta(P_i,\bfv_i)-\mu$. We note that the analytic mean over copies of ($P,\bfv$) for each $X_i$ is a zero vector. Furthermore, we obtain that 
    \begin{align}
        \sum_{i=1}^{N}\mathbb{E}(\|X_i\|_2^2)&\le N\sup_{i\in [N],(P_i,\bfv_i)}\left\{\|\delta(P_i,\bfv_i)-\mu\|_2^2\right\}\nonumber\\&\le N\left(1+\sum_{P}|q_P|\right)^2
    \end{align}
    Then by Cor. 1 in Ref.~\cite{pinelis1986}, we conclude that 
    {\small
    \begin{align}
         &P\left(\left\|\frac{1}{N}\sum_{i\in[N]}\delta(P_i,\bfv_i)-\mu\right\|_2\ge \frac{t}{N}\right)\nonumber\\&=P\left(\left\|\frac{1}{N}\sum_{i\in[N]}X_i\right\|_2\ge \frac{t}{N}\right)\nonumber\\
    &=P\left(\left\|\sum_{i\in[N]}X_i\right\|_2\ge t\right)\le 2e^{\frac{-t^2}{2N\left(1+\sum_{P}|q_P|\right)^2}},
    \end{align}}
    where $P(\cdot)$ means the probability to make $\cdot$ happen.
    
    Now, we assume $\frac{t}{N}=\epsilon$. Given that $N>M\cdot \mathcal{O}\left(\frac{\left(\sum_{P}|q_P|\right)^2}{\epsilon^2}\right)$ with sufficiently large constant $M$, then $t>1+\sum_P|q_P|$ and  $P\left(\left\|\frac{1}{N}\sum_{i=1}^{N}\delta(P_i,\bfv_i)-\mu\right\|_2>\epsilon\right)\le e^{-M\mathcal{O}(1)}$ which is negligible. 

    We can derive a similar result even if we sample $\bfu$ from the iterative convolution of $\mu$, following the same proof line. Specifically, we note that the following holds,

{\small
\begin{align}\label{eq:sampling_noisy_self_convolution}
    \mu^{\ast j}_{\bfu}&=\sum_{\substack{\bfu_1+\ldots+\bfu_{j+1}=\bfu\\\bfu_1,\ldots,\bfu_{j+1}\in \bfZ^n_2}}\left\{\sum_{P_1,\ldots,P_j\in \mathcal{P}_n}\left(\sum_{P\in \mathcal{P}_n}\frac{|q_P|}{\sum_{P\in \mathcal{P}_n}|q_P|}\right)^j\right.\nonumber\\&\left.\left(\sum_{P\in \mathcal{P}_n}|q_P|\right)^j\prod_{i=1}^j \left(\sgn(q_{P_i})\sum_{\bfv_i\in \bfZ^n_2}{\rm Prob}(\bfv_i|P_i)\delta_{\bfu_i,\bfv_i}\right)\right\}\nonumber\\&=\left(\sum_{P\in \mathcal{P}_n}|q_P|\right)^j\sum_{P_1,\ldots,P_j\in \mathcal{P}_n}\left(\sum_{P\in \mathcal{P}_n}\frac{|q_P|}{\sum_{P\in \mathcal{P}_n}|q_P|}\right)^j\nonumber\nonumber\\&\prod_{i=1}^j\left(\sgn(q_{P_i})\right)\sum_{\substack{\bfu_1+\ldots+\bfu_{j+1}=\bfu\\\bfu_1,\ldots,\bfu_{j+1}\in \bfZ^n_2}}\prod_{i=1}^{j}\left(\sum_{\bfv_i\in \bfZ^n_2}{\rm Prob}(\bfv_i|P_i)\delta_{\bfu_i,\bfv_i}\right)\nonumber\\&=\left(\sum_{P\in \mathcal{P}_n}|q_P|\right)^j\sum_{P_1,\ldots,P_j\in \mathcal{P}_n}\left(\sum_{P\in \mathcal{P}_n}\frac{|q_P|}{\sum_{P\in \mathcal{P}_n}|q_P|}\right)^j\nonumber\\&\prod_{i=1}^j\left(\sgn(q_{P_i})\right)\sum_{\substack{\bfu_1+\ldots+\bfu_{j+1}=\bfu\\\bfu_1,\ldots,\bfu_{j+1}\in \bfZ^n_2}}\prod_{i=1}^{j}\mu^{(P_i)}_{\bfu_i},
\end{align}
}where $\mu^{(P_i)}$ is defined to be the Born probability of measurement outcomes after the post-processing of $P_i$ operations.

Therefore, the algorithm is as follows. We sample $P_1,P_2,\ldots,P_j$ i.i.d by the distribution of $\frac{|q_P|}{\sum_{P}|q_P|}$, and then we measure each circuit with the post-Pauli operations $P_i\;(i\in[j])$ to obtain $\bfv_i$. Then the estimator $\delta^{(j)}$ will be
\begin{align}\label{eq:overhead_noisy_convolution}
    \delta^{(j)}_\bfu=\left(\sum_{P\in \mathcal{P}_n}|q_P|\right)^j\prod_{i=1}^j\left(\sgn(q_{P_i})\right)\delta_{\bfu,\sum_{i=1}^{j}\bfv_i}.
\end{align}
Following the same logic of the proof of the case of $j=1$, we obtain the generalized condition for $P\left(\left\|\frac{1}{N}\sum_{i}\delta^{(j)}(P_i,\bfv_i)-\mu^{\ast j}\right\|_2>\epsilon\right)<e^{-M\mathcal{O}(1)}$, which is $N>M\cdot \mathcal{O}\left(\frac{\left(\sum_{P}|q_P|\right)^{2j}}{\epsilon^2}\right)$ with a sufficiently large constant $M$. Considering the pure case ($\sum_{P}|q_P|=1$), we prove (iii). 

\section{Detection algorithm with noisy circuits}\label{method:detection_noisy_circuit}

This part introduces the probabilistic error cancellation (PEC) algorithm for noisy circuits and discusses the required sampling copies. It is a well-known quantum error mitigation  (QEM)~\cite{endo2018,sun2021,suzuki2022,tsubouchi2024} method. This section considers when the order of hypergraph state $k\le 3$. In reality, the Clifford circuits in our main scheme are also noisy. Even though its degree is weaker than that of hypergraph state preparation, if the desired accuracy is much lower than the noise of the Clifford circuit, this would harm the accuracy of noise detection. There is a way to solve this problem with the aid of an error mitigation technique~\cite{endo2018,tsubouchi2023}. Using the Pauli twirling technique~\cite{emerson2007}, we assume each two-qubit gate undergoes Pauli noise.

We remember that the tailoring channel to each copy can be merged into the post-processing $CZ$-operations depending on the second-copy measurement (see App.~\ref{app:gate_merging}). After that, the merged $CZ$ operations also depend on the measurement outcome $\bfv\in \bfZ_2^n$ on the second copy. We denote $\mathcal{U}_c$ as the channel of transversal CNOT and $\mathcal{U}_{t}^{(\bfa,\bfv)}$ as the following merged CZ operation after the second-copy measurement. Here, $\bfa=(\bfa_1,\bfa_2)\in \bfZ_2^{2n}$ is a uniformly sampled vector for tailoring. Assuming $X$-operations as free, we denote $\rho_\bfa:=X^{\bfa}\rho^{\otimes 2} X^{\bfa}$. Then we can divide the convolution circuit scheme as $\mathcal{U}_{t}^{(\bfa,\bfv)}\circ\mathcal{U}_{c}$. 

Now, we suppose (a fixed) $\mathcal{U}_c$ undergoes a noise channel $\mathcal{N}_c$, and $\mathcal{U}^{(\bfa,\bfv)}_{t}$ is with the noise channel $\mathcal{N}_t^{(\bfa,\bfv)}$. Then the final-noisy measurement distribution $\bar{\mu}$ is 

\begin{align}
    &\bar{\mu}_\bfu={\rm tr}\bigg\{\ket{+^{\bfu}}\bra{+^{\bfu}}\times\sum_{\bfv\in \bfZ_2^n}\mathcal{N}_t^{(\bfa,\bfv)}\circ\mathcal{U}^{(\bfa,\bfv)}_t\nonumber\\&\left((I\otimes \ket{\bfv}\bra{\bfv})\mathcal{N}_c\circ\mathcal{U}_c\left(\rho^{\otimes 2}_\bfa\right)(I\otimes \ket{\bfv}\bra{\bfv})\right)\bigg\}
\end{align}

Next, we decompose the inverse of noise channels into a quasi-probabilistic sum of free Pauli channels.
\begin{align}\begin{cases}
    \mathcal{N}^{-1}_c(\cdot)&=\sum_{P\in \mathcal{P}_n}q^{(c)}_{P}P(\cdot)P\\\mathcal{N}^{(\bfa,\bfv)-1}_t(\cdot)&=\sum_{P\in \mathcal{P}_n}q^{(\bfa,\bfv)}_{P}P(\cdot)P,
\end{cases}
\end{align}
where $\mathcal{P}_n$ is $n$-qubit Pauli group and $\forall q^{(c)}_P,q^{(\bfa,\bfv)}_P\in\mathbb{R}$ satisfies $\sum_{P\in \mathcal{P}_n}q^{(c)}_P=\sum_{P\in \mathcal{P}_n}q^{(c)}_P=1$. Other free channels are also preferred to improve the efficiency~\cite{piveteau2022}. Then the desired Born probability $\mu$ is expressed by 
{\small
\begin{align}
\mu_{\bfu}&=\frac{1}{4^n}\sum_{\bfa\in \bfZ^{2n}_n}{\rm tr}\bigg\{\ket{+^{\bfu}}\bra{+^{\bfu}}\sum_{\bfv\in \bfZ_2^n}\mathcal{N}_t^{(\bfa,\bfv)-1}\circ\mathcal{N}_t^{(\bfa,\bfv)}\circ\nonumber\\&\mathcal{U}^{(\bfa,\bfv)}_t\left((I\otimes \ket{\bfv}\bra{\bfv})\mathcal{N}_c^{-1}\circ\mathcal{N}_c\circ\mathcal{U}_c\left(\rho^{\otimes 2}_\bfa\right)\right)(I\otimes \ket{\bfv}\bra{\bfv})\bigg\}\nonumber\\&=\frac{1}{4^n}\sum_{\bfa\in \bfZ^{2n}_n}\sum_{P\in\mathcal{P}_n}\left|q^{(c)}_P\right|{\rm tr}\Bigg\{\ket{+^{\bfu}}\bra{+^{\bfu}}\sum_{\bfv\in \bfZ_2^n}\sum_{P\in\mathcal{P}_n}\left|q^{(\bfa,\bfv)}_P\right|\nonumber\\&\sum_{Q\in\mathcal{P}_n}\frac{\left|q^{(\bfa,\bfv)}_Q\right|\sgn(q^{(\bfa,\bfv)}_Q)}{\sum_{P\in \mathcal{P}_n}\left|q^{(\bfa,\bfv)}_P\right|}\mathcal{Q}\circ\mathcal{N}_t^{(\bfa,\bfv)}\circ\mathcal{U}^{(\bfa,\bfv)}_t\nonumber\\&\Bigg((I\otimes \ket{\bfv}\bra{\bfv})\sum_{R\in \mathcal{P}_n }\frac{\left|q^{(c)}_R\right|\sgn(q^{(c)}_R)}{\sum_{P\in \mathcal{P}_n}\left|q^{(c)}_P\right|}\mathcal{R}\circ\mathcal{N}_c\circ\mathcal{U}_c(\rho^{\otimes 2}_\bfa)\nonumber\\&(I\otimes \ket{\bfv}\bra{\bfv})\Bigg)\Bigg\}\nonumber
\end{align}

\begin{align}
&=\frac{1}{4^n}\sum_{P\in\mathcal{P}_n}\left|q^{(c)}_P\right|\sum_{\bfa\in \bfZ^{2n}_n}\sum_{R\in \mathcal{P}_n }\frac{\left|q^{(c)}_R\right|\sgn(q^{(c)}_R)}{\sum_{P\in \mathcal{P}_n}\left|q^{(c)}_P\right|}\nonumber\\&\sum_{\bfv\in \bfZ_2^n}P^{(c,R)}_{\rho_\bfa}(\bfv)\sum_{P\in\mathcal{P}_n}\left|q^{(\bfa,\bfv)}_P\right|\sum_{Q\in\mathcal{P}_n}\frac{\left|q^{(\bfa,\bfv)}_Q\right|\sgn(q^{(\bfa,\bfv)}_Q)}{\sum_{P\in \mathcal{P}_n}\left|q^{(\bfa,\bfv)}_P\right|}\times\nonumber\\&{\rm tr}\Bigg\{\ket{+^{\bfu}}\bra{+^{\bfu}}\mathcal{Q}\circ\mathcal{N}_t^{(\bfa,\bfv)}\circ\mathcal{U}^{(\bfa,\bfv)}_t\nonumber\\&\left(\frac{(I\otimes \ket{\bfv}\bra{\bfv})\mathcal{R}\circ\mathcal{N}_c\circ\mathcal{U}_c\left(\rho^{\otimes 2}_\bfa\right)(I\otimes \ket{\bfv}\bra{\bfv})}{P^{(c,R)}_{\rho_\bfa} (\bfv)}\right)\Bigg\},
\end{align}
}where $P_{\rho_\bfa}^{(c,R)}(\bfv)=\tr{(I\otimes \ket{\bfv}\bra{\bfv})\mathcal{R}\circ\mathcal{N}_c\circ\mathcal{U}_c(\rho_\bfa^{\otimes 2})}$, and $\mathcal{Q},\mathcal{R}$ are Pauli channel corresponding to Pauli operator $Q,R$ respectively. 

Finally, we set the mitigation algorithm as follows:
i) we uniformly sample $\bfa\in \bfZ_2^{2n}$, and a Pauli operator $R$ from $\left\{\frac{\left|q^{(c)}_R\right|}{\sum_{P\in \mathcal{P}_n}\left|q^{(c)}_P\right|}\right\}_{R\in \mathcal{P}_n}$. ii) We give the noisy channel $\mathcal{R}\circ \mathcal{N}_c\circ \mathcal{U}_c$ to $\rho^{\otimes 2}_\bfa$ and take the second-copy measurement to obtain the outcome $\bfv$. The sampling distribution of $\bfv$ follows $P_{\rho_\bfa}^{(c,R)}(\bfv)$. iii) We sample a Pauli operator $Q$ from $\left\{\frac{\left|q^{(c)}_Q\right|}{\sum_{P\in \mathcal{P}_n}\left|q^{(c)}_P\right|}\right\}_{Q\in \mathcal{P}_n}$. We can find the merged channel $\mathcal{U}_t^{(\bfa,\bfv)}$, and hence we enact the channel $\mathcal{Q}\circ\mathcal{N}_t^{(\bfa,\bfv)}\circ\mathcal{U}^{(\bfa,\bfv)}$ to the projected output $\frac{(I\otimes \ket{\bfv}\bra{\bfv})\mathcal{R}\circ\mathcal{N}_c\circ\mathcal{U}_c\left(\rho^{\otimes 2}\right)(I\otimes \ket{\bfv}\bra{\bfv})}{P^{(c,R)}_{\rho_\bfa} (\bfv)}$. iv) Measure the first copy to obtain $+^\bfu$. Then, we take the estimator as $\widehat{\mu}:=\sum_{P\in\mathcal{P}_n}\left|q^{(c)}_P\right|\times \sum_{P\in\mathcal{P}_n}\left|q^{(\bfa,\bfv)}_P\right|\sgn(q^{(c)}_R)\sgn(q^{(\bfa,\bfv)}_Q)\delta_{\bfu}$, where $\delta_\bfu$ is Kronecker delta distribution that outputs $1$ only for $\bfu$.

The sampling routine from the iterative convolution of $\mu$ is similar to the argument in App.~\ref{method:sampling_complexity}. This is because given $j\in \mathbb{N}$, 
\begin{align}
    \mu^{\ast j}=\left(\mathbb{E}\widehat{\mu}\right)^{\ast j}&=\mathbb{E}_1\widehat{\mu}\ast \ldots\ast \mathbb{E}_j\widehat{\mu}
\end{align}
where $\mathbb{E}_{1,2,\ldots,}$ means that each estimator $\widehat{\mu}$ is sampled independently by the previous routine. Effectively, the convolution applies to the Kronecker delta distribution. Upon denoting $\widehat{\mu}_{\bfv_i}=C_{R_i,Q_i,\bfa_i,\bfv_i}\delta_{\bfu_i}$ such that 
\begin{align}
    &C_{R_i,Q_i,\bfa_i,\bfv_i}\nonumber\\&:=\sum_{P\in\mathcal{P}_n}\left|q^{(c)}_P\right|\times \sum_{P\in\mathcal{P}_n}\left|q^{(\bfa_i,\bfv_i)}_P\right|\sgn(q^{(c)}_{R_i})\sgn(q^{(\bfa_i,\bfv_i)}_{Q_i}),
\end{align}
\begin{align}
    \mu^{\ast j}&=\mathbb{E}_{1,2,\ldots,j}\left(\prod_{i=1}^{j}C_{R_i,Q_i,\bfa_i,\bfv_i}\right)\delta_{\bfu_1}\ast \delta_{\bfu_2}\ast\cdots\ast\delta_{\bfu_j}\nonumber\\&=\mathbb{E}_{1,2,\ldots,j}\left(\prod_{i=1}^{j}C_{R_i,Q_i,\bfa_i,\bfv_i}\right)\delta_{\bfu_1+\bfu_2+\cdots+\bfu_j}. 
\end{align}
With the same logic of  App.~\ref{method:sampling_complexity}, we obtain that required sampling for $l_2$ approximation is scaled by the square of the maximal-scale factor $\left(\sum_{P\in\mathcal{P}_n}\left|q^{(c)}_P\right|\times \max_{\bfa,\bfv\in \bfZ_2^n}\sum_{P\in\mathcal{P}_n}\left|q^{(\bfa,\bfv)}_P\right|\right)^j$.

\bibliography{sn-bibliography}

\end{document}